\documentclass{llncs}
\usepackage{amsmath,amssymb,epsfig}
\usepackage{eucal} % curlier script letters are easier to distinguish from italic
\usepackage{graphicx}
\usepackage{hyperref}
\usepackage{cite}
\usepackage{microtype}

\newif\iffull\fulltrue
\iffull
\usepackage[margin=1in]{geometry}
\fi

\usepackage{etoolbox}
{\makeatletter\patchcmd{\@spthm}{\phantomsection}{}{}{}}

\newtheorem{observation}{Observation}

\let\doendproof\endproof
\renewcommand\endproof{~\hfill$\qed$\doendproof}

\title{Reconfiguring Undirected Paths}
\author{%
  Erik D. Demaine%
    \thanks{MIT Computer Science and Artificial Intelligence Laboratory,
      32 Vassar St., Cambridge, MA 02139, USA, \protect\url{edemaine@mit.edu}}
\and
  David Eppstein%
    \thanks{Computer Science Department, University of California, Irvine, Irvine, CA 92697, USA, \protect\url{eppstein@uci.edu}. Supported in part by NSF grants CCF-1618301 and CCF-1616248.}
\and
  Adam Hesterberg%
    \thanks{MIT Mathematics Department,
      77 Massachusetts Ave., Cambridge, MA 02139, USA, \protect\url{achesterberg@gmail.com}}
\and
  Kshitij Jain%
  \thanks{University of Waterloo,
	Canada, \protect\url{k22jain@uwaterloo.ca}}
\and
\\Anna Lubiw%
  \thanks{University of Waterloo,
	Canada, \protect\url{alubiw@uwaterloo.ca}}
\and
 Ryuhei Uehara%
   \thanks{Japan Advanced Institute of Science and Technology, Nomi, Japan, \protect\url{uehara@jaist.ac.jp}}
\and
  Yushi Uno%
   \thanks{Graduate School of Engineering, Osaka Prefecture University, Japan, \protect\url{uno@cs.osakafu-u.ac.jp }}
}
\institute{}

\begin{document}
\maketitle

\begin{abstract}
We consider problems in which a simple path of fixed length, in an undirected graph, is to be shifted from a start position to a goal position by moves that add an edge to either end of the path and remove an edge from the other end. We show that this problem may be solved in linear time in trees, and is fixed-parameter tractable when parameterized either by the cyclomatic number of the input graph or by the length of the path. However, it is $\mathsf{PSPACE}$-complete for paths of unbounded length in graphs of bounded bandwidth.
\end{abstract}

\section{Introduction}

In this paper, we consider the problem of sliding a fixed-length simple path within an undirected graph
from a given starting position to a given goal position. The path may move in steps where we add an edge to either end of the path and simultaneously remove the edge from the opposite end, maintaining its length. Effectively, this can be thought of as sliding the path one step along its length in either direction.
The allowed movements of the path are similar to those of trains in a switchyard, or of the model trains in any of several train shunting puzzles; the edges of the path can be thought of as the cars of a train. However, unlike train tracks, we do not constrain connections at junctions of track segments to be smooth: a path that enters a vertex along an incident edge can exit the vertex along any other incident edge. Additionally, we do not distinguish the two ends of the path from each other.

\begin{figure}[t]
\centering\includegraphics[scale=0.45]{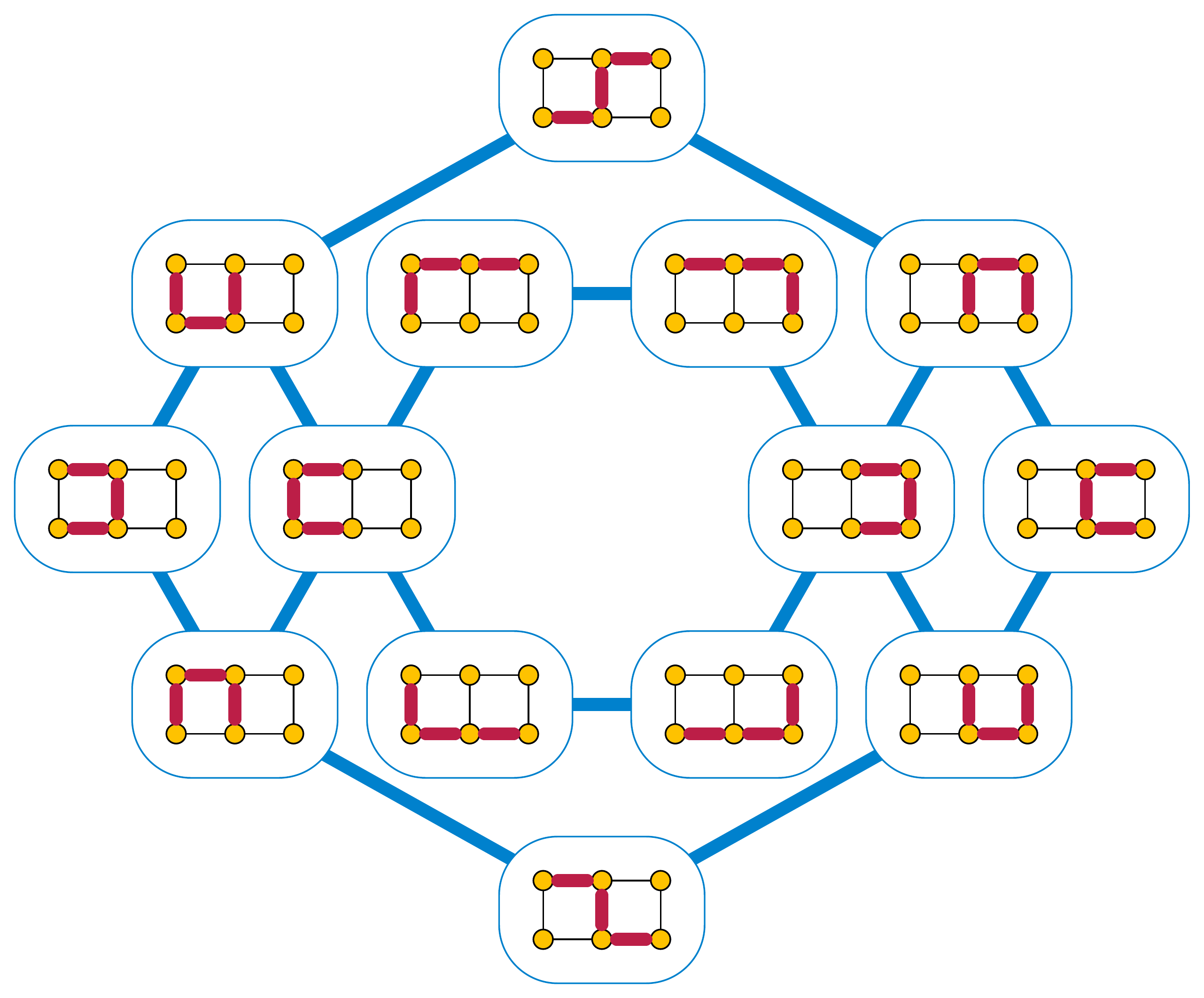}
\caption{State space of three-edge paths on a six-vertex graph}
\label{fig:statespace}
\end{figure}

Our aim is to understand the computational complexity of two natural reconfiguration problems for such paths: the \emph{decision problem}, of testing whether it is possible to reach the goal position from the start position, and the \emph{optimization problem}, of reaching the goal from the start in as few moves as possible. One natural upper bound for the complexity of these problems is the size of the state space for the problem, a graph whose vertices are paths of equal length on the given graph and whose edges represent moves from one path to another (\autoref{fig:statespace}).
If a given graph has $N$ paths of the given length, and $M$ moves from one path to another, we can solve either the decision problem or the optimization problem in time $O(M+N)$ (after constructing the state space) by a simple breadth-first search. As we will see, it is often possible to achieve significantly faster running times than this naive bound. On the other hand, the general problem is hard, even on some highly restricted classes of graphs.

Specifically, we prove the following results:
\begin{enumerate}
\item The decision problem for path reconfiguration is fixed-parameter tractable when parameterized by the length of the path. This stands in contrast to the size of the state space for the problem which (for paths of length $k$ in $n$-vertex graphs) can have as many as $\Omega(n^{k+1})$ states.
\item For paths of unbounded length in graphs parameterized by the circuit rank, both the decision and the optimization problems can be solved in fixed-parameter tractable time by state space search. The same problem can be solved in polynomial (but not fixed-parameter tractable) time when parameterized by feedback vertex set number.
\item The optimization problem for path reconfiguration in trees can be solved in linear time, even though the state space for the problem has quadratic size.
\item The decision problem for path reconfiguration is $\mathsf{PSPACE}$-complete for paths of unbounded length, even when restricted to graphs of bounded bandwidth. Therefore (unless $\mathsf{P} = \mathsf{PSPACE}$) path reconfiguration is not fixed-parameter tractable when parameterized by bandwidth, treewidth, or related graph parameters.
\end{enumerate}

\iffull
\else
Because of limited space, the detailed versions of several of our results are deferred to the full version of the paper.
\fi

\subsection{Related work}

There has been much past research on reconfiguring structures in graphs,
with motivations that include motion planning, understanding the mixing of Markov chains and bounding the computational complexity of popular games and puzzles. See, for instance, Ito et al.~\cite{ItoDemHar-TCS-11} for many early references, and 
Mouawad et al.~\cite{MouNisRam-Algo-17} for more recent work on the parameterized complexity of these problems. Often, in these problems, one considers moves in which the structure changes by the removal of one element and the addition of an unrelated replacement element (token moving) or in which an element of the structure changes only locally, by moving along an edge of the graph (token sliding). 

Several authors have considered problems of reconfiguring paths or shortest paths under token jumping or token sliding models of reconfiguration~\cite{KamMedMil-TCS-11,Bon-TCS-13,HanItoMiz-COCOON-18}. However, the path sliding moves that we consider are different.
Token sliding moves only a single vertex or edge of a path along a graph edge, while we move the whole path. And although our path sliding moves can be seen as a special case of token jumping, because they remove one edge and add a different edge, token jumping in general would allow the replacement of edges or vertices in the middle of a path, while we allow changes only at the ends of the path.

The path reconfiguration problem that we study here is also closely related to a popular video game, Snake, which has a very similar motion to the path sliding moves that we consider. Our problem differs somewhat from Snake in that we consider bidirectional movement, while in Snake the motion must always be forwards. Snake is typically played on grid graphs, and it is known to be $
\mathsf{PSPACE}$-complete to determine whether the Snake can reach a specific goal state from a given start state on generalized grid graphs~\cite{DeBOph-FUN-16}. Independently of our work,
Gupta et al~\cite{GupSaaZeh-arXiv-19} have found that reconfiguring snakes (paths that can move only unidirectionally) is fixed-parameter tractable in the length of the path, analogously to our \autoref{thm:fpt-in-path-length}.

\section{Preliminaries}

\subsection{Reconfiguration sequences and time reversal}

\begin{definition}
We define a \emph{reconfiguration step} in a graph $G$ to be a pair of edges $(e,f)$, and a \emph{reconfiguration sequence} to be a sequence $\sigma$ of reconfiguration steps.
We may \emph{apply} a reconfiguration step to a path $P$ by adding edge $e$ to $P$ and removing edge $f$, whenever $f$ is one of the two edges at the ends of $P$, $e$ is incident to the vertex at the other end, and the result of the application is another simple path. We may apply a reconfiguration sequence to a path by performing a sequence of applications of its reconfiguration steps.
If applying reconfiguration sequence $\sigma$ to path $P$ produces another path $Q$ we say that we can reconfigure $P$ into $Q$ or that $\sigma$ takes $P$ to $Q$.

If $(e,f)$ is a reconfiguration step, then we define its \emph{time reversal} to be the step $(f,e)$.
We define the time-reversal of a reconfiguration sequence $\sigma$ to be the sequence
of time reversals of the steps of $\sigma$, taken in the reverse order.
If $\sigma$ takes $P$ to $Q$, then its time reversal takes $Q$ to $P$.
For this reason, when we seek the existence of a reconfiguration sequence (the path reconfiguration decision problem) or the shortest reconfiguration sequence (the path reconfiguration optimization problem), reconfiguring a path $P$ to $Q$ is equivalent under time reversal to reconfiguring $Q$ to $P$.
We call this equivalence \emph{time-reversal symmetry}.

We define the \emph{length} $|P|$ of a path $P$ to be its number of edges, and the length~$|\sigma|$ of a reconfiguration sequence to be its number of steps.
\end{definition}

\subsection{Tree-depth}

\begin{figure}[t]
\centering\includegraphics[scale=0.5]{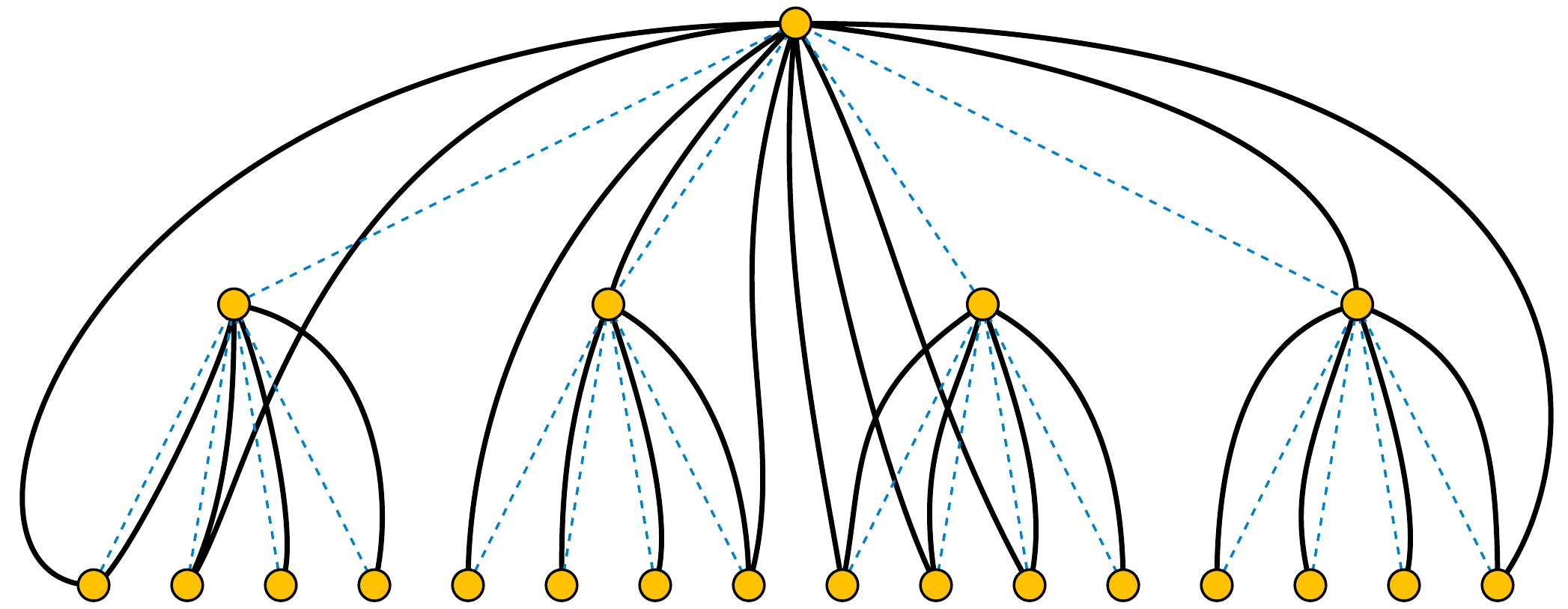}
\caption{A graph $G$ of tree-depth 2 (solid black edges) and a tree $T$ realizing this depth (dashed blue edges).}
\label{fig:treedepth}
\end{figure}

Tree-depth is a graph parameter that can be defined in several equivalent ways~\cite{NesOss-12}, but the most relevant definition for us is that the tree-depth of a connected graph $G$ is the minimum depth of a rooted tree $T$ on the vertices of $G$ such that each edge of $G$ connects an ancestor-descendant pair of $T$ (\autoref{fig:treedepth}). Here, the depth of a tree is the length of the longest root-to-leaf path. Another way of expressing the connection between $G$ and $T$ is that $T$ is a depth-first search tree for a supergraph of~$G$. For disconnected graphs one can use a forest in place of a tree, but we will only consider tree-depth for connected graphs.

Tree-depth is a natural graph parameter to use for path configuration, because it is closely connected to the lengths of paths in graphs. If a graph $G$ has maximum path-length $\ell$, then clearly its tree-depth can be at most $\ell$, because any depth-first search tree of $G$ itself will achieve that depth. In the other direction, a graph with tree-depth $d$ has maximum path-length at most $2^{d+1}-2$, as can be proven inductively by splitting any given path at the vertex closest to the root of a tree~$T$ realizing the tree-depth. Therefore, the tree-depth and maximum path-length are equivalent for the purposes of determining fixed-parameter tractability. The parameterized complexity of reconfiguration problems on graphs of bounded tree-depth has been studied by Wrochna~\cite{Wro-JCSS-18}.
However, these graphs are highly constrained, so algorithms that are parameterized by tree-depth are not widely applicable.

We will prove as a lemma that path reconfiguration is fixed-parameter tractable for the graphs of bounded tree-depth. Because these graphs have bounded path lengths, this result will be subsumed in our theorem that path reconfiguration is fixed-parameter tractable when parameterized by path-length. However, we will use this lemma as a stepping-stone to the theorem, by proving that in arbitrary graphs we can either find a structure that allows us to solve the problem easily or restrict the input to a subgraph of bounded tree-depth.

\section{Parameterized by path length}

In this section we show that path reconfiguration is fixed-parameter tractable when parameterized by path length. As discussed above, our strategy is to find a structure (\emph{loose paths}, defined below), whose existence allows us to solve the reconfiguration problem directly. When these structures do not exist or exist but cannot be used, we will instead restrict our attention to a subgraph of bounded tree-depth. We begin with the lemma that the problem is fixed-parameter tractable when parameterized by tree-depth instead of path length.

\subsection{Tree-depth}

Our method for graphs of low tree-depth is based on the fact that, when these graphs are large, they contain a large amount of redundant structure: subgraphs that are all connected to the rest of the graph in the same way as each other. When this happens, we can eliminate some copies of the redundant structures and reduce the problem to a smaller instance size.

\begin{definition}
Given a graph $G$ and a vertex set $S$, we define an \emph{$S$-flap} to be a subset $X$ of the vertices of $G$ such that $X$ is disjoint from $S$ and there are no edges from $X$ to $G\setminus\{S\cup X\}$.
We say that two $S$-flaps $X$ and $Y$ are \emph{equivalent} when the induced subgraphs $G[S\cup X]$ and $G[S\cup Y]$ are isomorphic, by an isomorphism that reduces to the identity mapping on $S$ (\autoref{fig:equiv-flaps}).
\end{definition}

\begin{observation}
\label{obs:too-many-flaps}
For any graph $G$ and any vertex set $S$,
a path of length $k$ can include vertices from at most $\lceil (k-1)/2\rceil$ $S$-flaps of $G$.
\end{observation}

\begin{figure}[t]
\centering\includegraphics[scale=0.5]{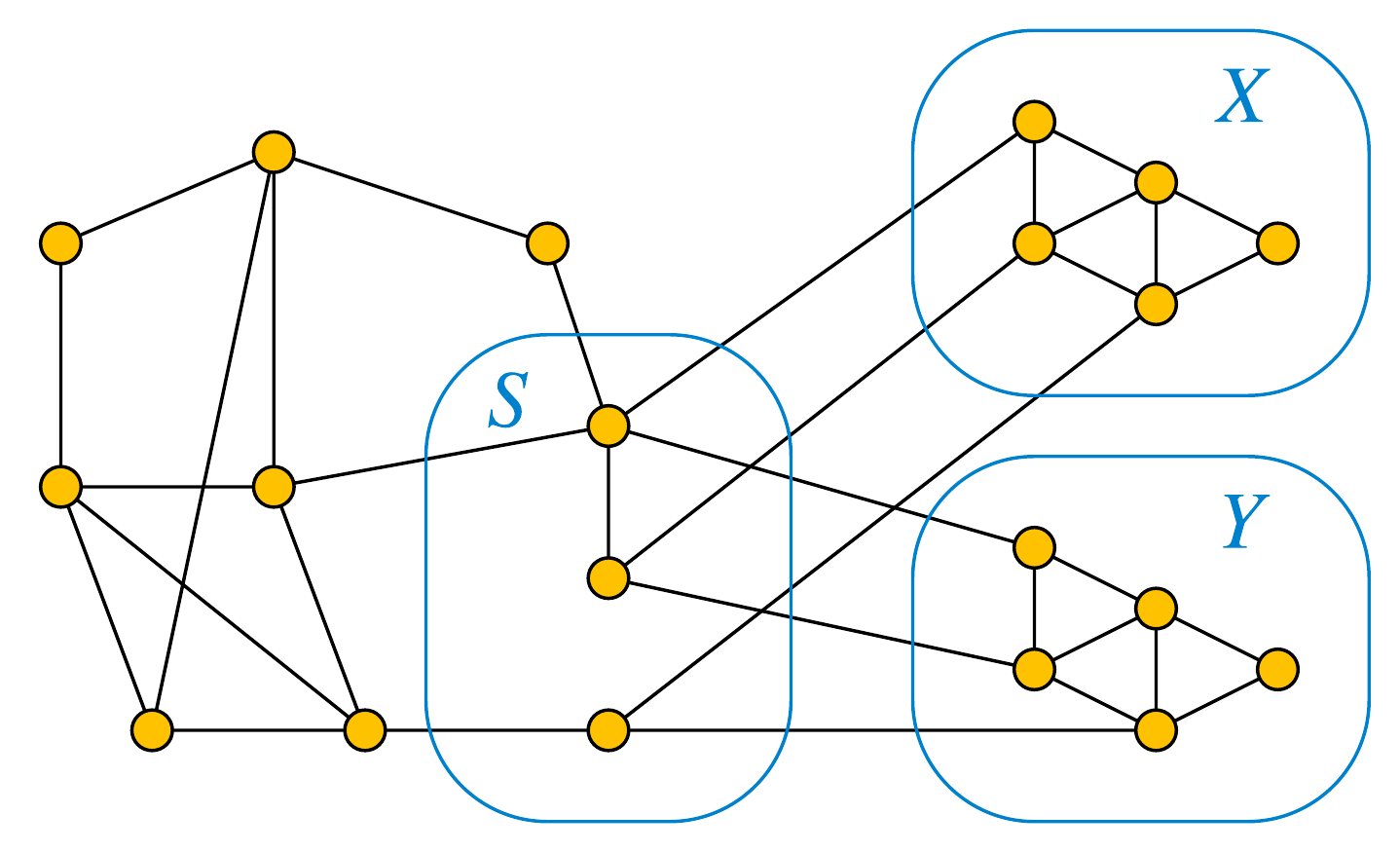}
\caption{Two equivalent $S$-flaps $X$ and $Y$ in a graph $G$}
\label{fig:equiv-flaps}
\end{figure}

\begin{proof}
The path has $k+1$ vertices, and any two vertices in distinct flaps must be separated by at least one vertex of $S$.
\end{proof}

\begin{lemma}
\label{lem:remove-flaps}
Suppose we are given an instance of path reconfiguration for paths of length $k$ in a graph $G$,
and that $G$ contains a subset $S$ that is disjoint from the start and goal positions of the path
and has more than $\lceil (k+1)/2\rceil$ pairwise equivalent $S$-flaps $X_1,X_2,\dots$, all disjoint from the start and goal. Then we can construct an equivalent and smaller instance by removing all but $\lceil (k+1)/2\rceil$of these equivalent $S$-flaps.
\end{lemma}

\begin{proof}
Any reconfiguration sequence in the original graph can be transformed into a reconfiguration sequence for the reduced graph by using one of the remaining $S$-flaps whenever the sequence for the original graph enters an $S$-flap. Because the $S$-flaps are equivalent, the moves within the flap can be mapped to each other by the isomorphism defining their equivalence, and by \autoref{obs:too-many-flaps} there will always be a free $S$-flap to use in the reduced graph.
\end{proof}

\begin{lemma}
\label{lem:fpt-tree-depth}
We can solve the decision or optimization problems for path reconfiguration in time that is fixed-parameter tractable in the tree-depth of the input graph.
\end{lemma}

\begin{proof}
We provide a polynomial-time kernelization algorithm that uses \autoref{lem:remove-flaps} to reduce the instance to an equivalent instance whose size is a function only of the given tree-depth~$d$.
The problem can then be solved by a brute-force search on the resulting smaller instance.
We assume without loss of generality that we already have a tree decomposition $T$ of depth~$d$, as it is fixed-parameter tractable to find such a decomposition when one is not already given~\cite[p.~138]{NesOss-12}. Recall that, for graphs of tree-depth $d$, the length $k$ of the paths being reconfigured can be at most $2^{d+1}-2$.

We apply \autoref{lem:remove-flaps} in a sequence of stages so that, after stage~$i$, all vertices at height $i$ in $T$ have $O(1)$ children. As a base case, for stage~0, all vertices at height 0 in $T$ automatically have 0 children, because they are the leaves of $T$.
Therefore, suppose by induction on $i$ that all vertices at height less than $i$ in $T$ have $O(1)$ children.

For a given vertex $v$ at height~$i$, let $S_v$ be the set of ancestors of $v$ in $T$ (including $v$ itself).
Then, for each child $w$ of $v$ in $T$, let $X_w$ be
the set of descendants of $w$ (including $w$ itself). Then $X_w$ is an $S_v$-flap, because $S_v$ includes all of its ancestors in $T$ and it can have no edges to vertices that are not ancestors in $T$.
If we label each vertex in $T$ by the set of heights of its adjacent ancestors,
then the isomorphism type of $G[S_v\cup X_w]$ is determined by these labels,
so two children $u$ and $w$ of $T$ have equivalent $S_v$-flaps whenever
they correspond to isomorphic labeled subtrees of $W$.
Trees of constant size with a constant number of label values can have a constant number of isomorphism types, so there are a constant number of equivalence classes of $S_v$ flaps among the sets $W_x$. Within each equivalence class, we apply \autoref{lem:remove-flaps} to reduce the number of flaps within that equivalence class to a constant. After doing so, we have caused the vertices of $T$ at height $i$ to have a constant number of children, completing the induction proof.

To implement this method in polynomial time, we can use any polynomial time algorithm for isomorphism of labeled trees~\cite{HopWon-STOC-74}. The equivalence of subtrees of $T$ by labeled isomorphism may be finer than the equivalence of the corresponding subgraphs of $G$ by graph isomorphism (because two different labeled trees may correspond to isomorphic subgraphs) but using the finer equivalence relation nevertheless leaves us with a kernel of size depending only on~$d$. The time for this algorithm can be bounded by a polynomial, independent of the parameter.
\end{proof}

As the following observation shows, this result is nontrivial in the sense that its time bound is significantly smaller than the worst-case bound on the size of the state space for the problem.

\begin{observation}
In graphs of tree-depth $d$, the number of paths of a given length can be $\Theta(n^{2^d})$.
\end{observation}

\begin{figure}
\centering\includegraphics[scale=0.5]{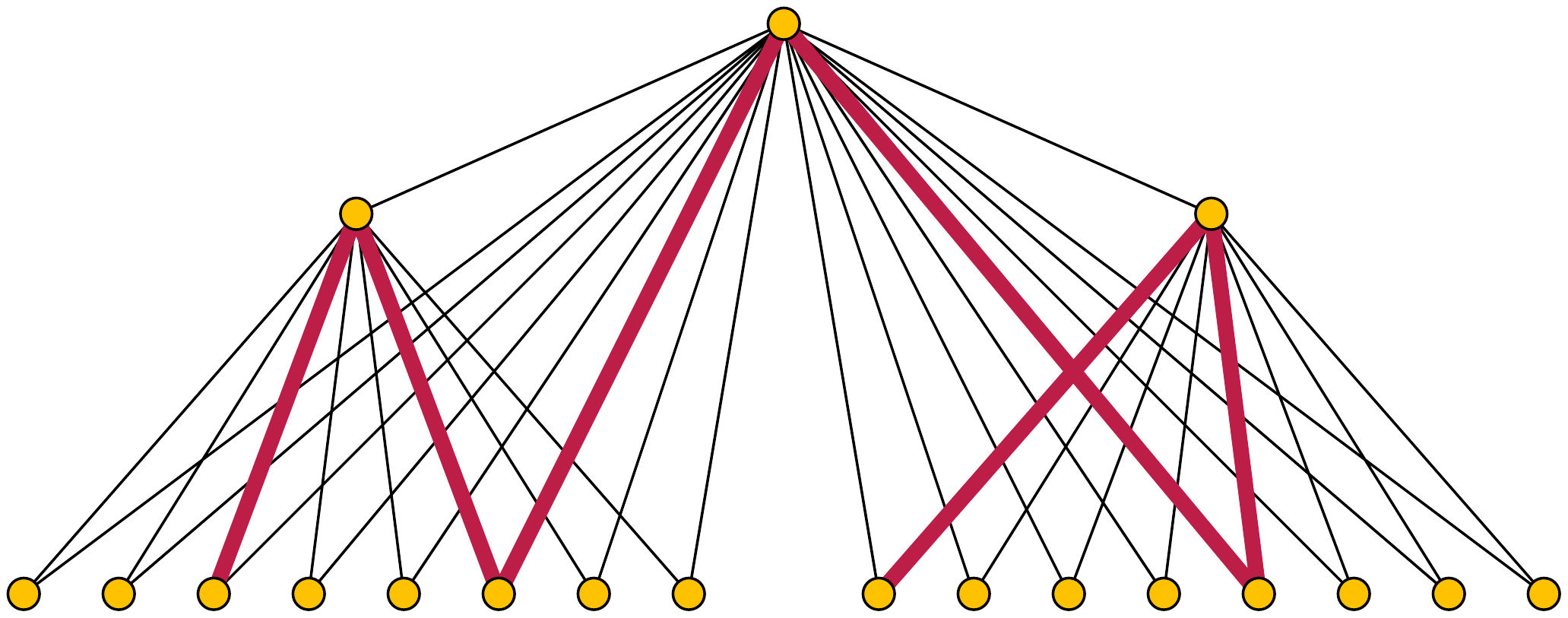}
\caption{One of $\Omega(n^4)$ paths of length 6 in a graph of tree-depth 2}
\label{fig:many-paths}
\end{figure}

\begin{proof}
Let $T$ be a tree realizing the depth of the given graph.
To prove that the number of paths is $O(n^{2^d})$, consider the vertex $v$ in any path that is highest in tree~$T$, and apply the same bound inductively for the two parts of the path on either side of $v$, both of which must live in lower-depth subtrees. The total number of paths can be at most the product of the numbers of choices for these two smaller paths.

To prove that the number of paths can be $\Omega(n^{2^d})$, let $T$ be a star as the base case for depth one (with $\Omega(n^2)$ paths of length two) and at each higher depth connect two inductively-constructed subtrees through a new root vertex $v$. Given a tree $T$ constructed in this way, let $G$ be the graph of all ancestor-descendant pairs in $T$ (\autoref{fig:many-paths}). Each two paths in the two subtrees can be connected to each other through $v$, so the number of paths in the whole graph is the product of the numbers of paths in the two subtrees.
\end{proof}

Therefore, an algorithm that searched the entire state space would only be in $\mathsf{XP}$, not $\mathsf{FPT}$.

\subsection{Loose paths}

We have seen that graphs without long paths are easy for path reconfiguration. Next, we show that graphs with long paths are also easy. The following definition is central to this part of our results:

\begin{figure}[t]
\centering\includegraphics[scale=0.45]{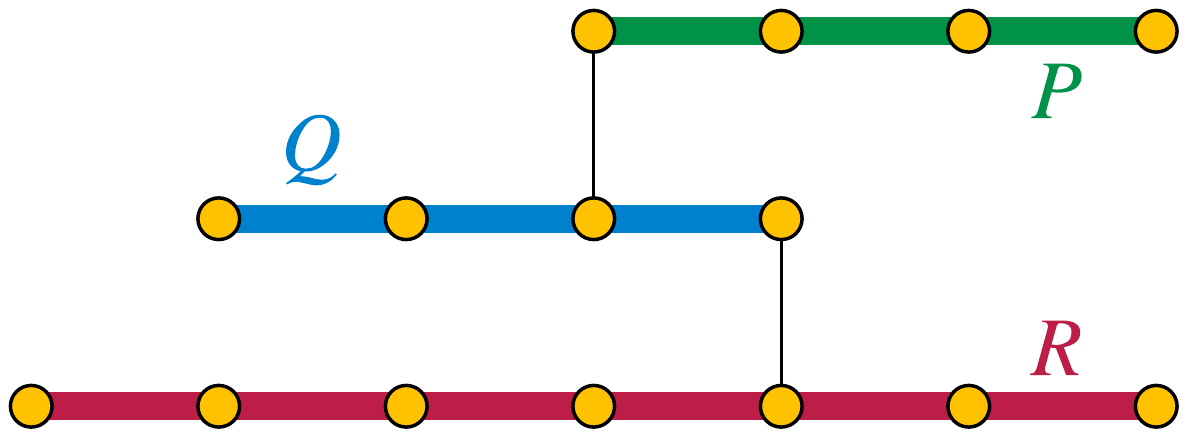}
\caption{A loose path $R$ for start and goal paths $P$ and $Q$}
\label{fig:loose-path}
\end{figure}

\begin{definition}
Consider an instance of path reconfiguration consisting of a graph $G$, a start path $P$ of length $k$, and a goal path $Q$ of length $k$. We define a \emph{loose path} to be a simple path $R$ of length $2k$ in $G$, such that $R$ is vertex-disjoint from both $P$ and $Q$ (\autoref{fig:loose-path}).
\end{definition}

\begin{lemma}
\label{lem:reachable-loose-path}
Let $R$ be a loose path for an instance $(G,P,Q)$ of path reconfiguration, such that it is possible to reconfigure path $P$ into a path that uses at least one vertex of $R$. Then for every vertex $v$ in $R$, it is possible to reconfigure path $P$ into a sub-path of $R$ for which $v$ is an endpoint.
\end{lemma}

\begin{proof}
Consider a sequence $\sigma$ of reconfiguration steps starting from $P$ that results in a path using at least one vertex of $R$ and is as short as possible. Because $\sigma$ is as short as possible and $R$ is disjoint from $P$, the last move of $\sigma$ must cause exactly one vertex $u$ of $R$ to be an endpoint of the reconfigured path. Because $R$ has length $2k$, at least one endpoint of $R$ is at distance $k$ or more along $R$ from $u$. By sliding the path along $R$ towards this endpoint, we can reconfigure it so that it lies entirely along $R$. Again, because $R$ has length $2k$, one of the two sub-paths of $R$ ending at $v$ has length at least $k$.  By concatenating to $\sigma$ an additional sequence of steps that slide the path along $R$ (if necessary) we can reconfigure the starting path so that it lies within this sub-path and ends at $v$. 
\end{proof}

We call a loose path $R$ that meets the conditions of \autoref{lem:reachable-loose-path} a \emph{reachable loose path}.

\begin{lemma}
\label{lem:all-loose-paths}
If an instance $(G,P,Q)$ of path reconfiguration has a reachable loose path, and the graph $G$ is connected, then all loose paths for that instance are reachable.
\end{lemma}

\begin{proof}
Let $R$ be a reachable loose path, and $L$ be any other loose path.
If $R$ and $L$ share a vertex $v$, then it is possible to slide any sub-path of $R$ so that it includes this vertex, showing that $L$ meets the conditions of \autoref{lem:reachable-loose-path}.
If $R$ and $L$ are disjoint, let $T$ be a shortest path between them in $G$, and let $v$ be the unique vertex of $T$ that belongs to $R$. By \autoref{lem:reachable-loose-path}, we can reconfigure the starting path so that it lies along $R$ and ends at~$v$.
From there, we can slide the path along $T$ until it reaches the other endpoint of $T$, a vertex of $L$. This shows that $L$ meets the conditions of \autoref{lem:reachable-loose-path}.
\end{proof}

It will be helpful to bound the tree-depth of graphs with no loose path.

\begin{observation}
If an instance of path reconfiguration for paths of length $k$ has no loose path,
then its graph has tree-depth less than $4k$.
\end{observation}

\begin{proof}
Form a depth-first-search forest $F$ of the subgraph formed by removing all vertices of the start and goal paths. Because there is no loose path, $F$ has depth at most $2k-1$. Form a single rooted path $R$ of the vertices of the start and goal paths, in an arbitrary order. Connect $R$ and $F$ into a single tree $T$ (not necessarily a subtree of the input graph) by making each root of $F$ be a child of the leaf node of $R$. Then every edge in the given graph connects an ancestor--descendant pair in $T$, because either it connects two vertices in the depth-first-search forest or it has at least one endpoint on the ancestral path $R$. Thus, $T$ meets the condition for trees realizing the tree-depth of a graph, and its depth is at most $4k-1$, so the given graph has tree-depth at most $4k-1$.
\end{proof}

\subsection{Win-win}

We show now that we can either restrict our attention to a subgraph of bounded tree-depth or find a reachable loose path, in either case giving a structure that allows us to solve path reconfiguration.

\begin{definition}
Given an instance $(G,P,Q)$ of path reconfiguration, we say that $S$ is a \emph{reachable} set of vertices if, for every vertex $v$ in $S$, there exists a sequence of reconfiguration steps that takes $P$ into a path that uses vertex $v$. We say that $S$ is an \emph{inescapable} set of vertices if,
for every vertex $v$ that is not in $S$, there does not exist a sequence of reconfiguration steps that takes $P$ into a path that uses vertex $v$.
\end{definition}

\begin{lemma}
\label{lem:win-win}
Given an instance $(G,P,Q)$ of path reconfiguration, parameterized by the length $k$ of the start and goal paths, we can in fixed-parameter-tractable time find either a reachable loose path, or a reachable and inescapable set $S$ of vertices that induces a subgraph $G[S]$ of tree-depth at most $4k-1$.
\end{lemma}

\begin{proof}
We will maintain a vertex set $S$ that is reachable and induces a subgraph of tree-depth less than $4k$ until either finding reachable loose path or finding that $S$ is inescapable and has no path. Initially, $S$ will consist of all vertices of the start path $P$; clearly, this satisfies the invariants that $S$ is reachable and has tree-depth less than $4k$.

Then, while we have not terminated the algorithm, we perform the following steps:
\begin{itemize}
\item For each edge $uv$ where $u\in S$ and $v\not\in S$, use the algorithm of \autoref{lem:fpt-tree-depth} to test whether $P$ can be reconfigured within $S\cup\{v\}$ (a graph of tree-depth at most $4k$) into a path that uses vertex $v$. If we find any single edge $uv$ for which this test succeeds, we go on to the next step. Otherwise, if no edge $uv$ passes this test, $S$ is inescapable and we terminate the algorithm.
\item Test whether the graph $S\cup\{v\}$ contains a loose path. Finding a path of fixed length is fixed-parameter tractable for arbitrary graphs~\cite{Bod-Algs-93,FelLan-JCSS-94,AloYusZwi-JACM-95} and can be solved even more easily by standard dynamic programming techniques for graphs of bounded tree-depth. If this test succeeds, the loose path must contain $v$, as the remaining vertices have no loose path. In this case, we have found a reachable loose path (as $v$ is reachable) and we terminate the algorithm.
\item Add $v$ to $S$ and continue with the next iteration of the algorithm. Because (in this case) $v$ is reachable but $S\cup\{v\}$ contains no loose path, it follows that including $v$ in $S$ maintains the invariants that $S$ be reachable and induce a subgraph with tree-depth at most $4k-1$.
\end{itemize}
Because each iteration adds a vertex to $S$, the loop must eventually terminate, either with a reachable inescapable subgraph of low tree-depth (from the first step) or with a reachable loose path (from the second step).
\end{proof}

\subsection{Fixed-parameter tractability}

We are now ready to prove our main result:

\begin{theorem}
\label{thm:fpt-in-path-length}
The path reconfiguration decision problem is fixed-parameter tract\-able when parameterized by the length of the start and goal paths.
\end{theorem}

\begin{proof}
Our algorithm for path reconfiguration begins by applying \autoref{lem:win-win} to find either a reachable inescapable subgraph of low tree-depth or a reachable loose path.
If we find a reachable inescapable subgraph that does not include all the goal path vertices, the reconfiguration problem has no solution. If we find a reachable inescapable subgraph that does  include all the goal path vertices, we can solve the reconfiguration problem by applying \autoref{lem:fpt-tree-depth}.

If we find a reachable loose path $R$ for the given instance $(G,P,Q)$, we apply \autoref{lem:win-win} a second time, to the equivalent reversed instance $(G,Q,P)$. If we find a reachable inescapable subgraph that does not include all the vertices of the original start path $P$, the reconfiguration problem has no solution. If we find a reachable inescapable subgraph that does  include all the vertices of $P$, we can solve the reconfiguration problem by applying \autoref{lem:fpt-tree-depth}. 

If we find a second reachable loose path $R'$, one that (by time-reversal symmetry) can reach the goal configuration, then the original reconfiguration problem has a positive solution. For, in this case, we can reconfigure $P$ to a path that lies along $R$, then (by \autoref{lem:all-loose-paths}) to a path that lies along $R'$, then (by the reverse of the reconfiguration sequence found by the second instance of \autoref{lem:win-win}) to $Q$.
\end{proof}

We leave as open the question of whether a similar result can be obtained for the optimization problem.

%%%%%%%%%%%%%%%%%%%%%%%%%%%%%%%%

\section{Tree-like graphs}
\label{sec:treelike}

\iffull
In this section we show that several special classes of graphs have polynomial algorithms for path reconfiguration regardless of path length. The prototypical example are the trees, for which the existence of a polynomial time algorithm follows immediately from the fact that any $n$-vertex tree has $O(n^2)$ distinct paths. In \autoref{sec:trees}, we refine this idea and provide a linear time algorithm for path reconfiguration in trees. In this section, we consider more general classes of graphs. For simplicity of exposition, rather than attempting to optimize the exponents of our algorithms (as we do for trees), we limit our work in this direction to determining which classes of graphs have polynomial-time or fixed-parameter tractable algorithms.

\subsection{Circuit rank}

The \emph{circuit rank} of an undirected graph $G$ is the minimum number $r$ such that $G$ may be decomposed into the edge-disjoint union of a forest $F$ and an $r$-edge graph $R$. If $G$ is connected, with $m$ edges and $n$ vertices, then we have the simple formula $r=m-n+1$. The forest $F$ may be chosen as any spanning forest of $G$; the remaining set $R$ of edges in $G\setminus F$ will automatically have $|R|=r$.

\begin{lemma}
Let $G$ be a connected graph formed as the edge-disjoint union of a forest $F$ and another graph $R$. Then the paths in $G$ are uniquely determined by their pairs of endpoints and by the intersection of their edge sets with $R$.
\end{lemma}

\begin{proof}
Suppose for a contradiction that there existed two distinct paths $P$ and $Q$ with the same two endpoints and the same intersections with $R$. Then the symmetric difference of their edge sets, $P\triangle Q$, would be a non-empty subgraph of $F$ with even degree at every vertex, contradicting the assumption that $F$ is a forest.
\end{proof}

\begin{corollary}
In an $n$-vertex graph with circuit rank $r$, there can be at most $2^r\tbinom{n}{2}$ distinct paths.
\end{corollary}

\begin{corollary}
Path reconfiguration (in both the decision and optimization problems) can be solved by breadth-first search of the state space in time that is fixed-parameter tractable in the circuit rank.
\end{corollary}

\subsection{Feedback vertex set}

A \emph{feedback vertex set} in a graph $G$ is a subset of vertices the removal of which would leave a forest. The \emph{feedback vertex set number} of $G$ is the minimum size of a feedback vertex set. (Analogously, the circuit rank can be thought of as the feedback edge set number.)

\begin{lemma}
In an $n$-vertex graph with feedback vertex set number $\phi$, the number of paths is at most
\[
\phi!\, 2^\phi \left( \tbinom{n-\phi}{2}+(n-\phi)+1 \right)^{\phi+1}
\]
\end{lemma}

\begin{proof}
A path can be completely specified by which vertices of the feedback vertex set are present, in what order these vertices are present along the path, and what path in the remaining forest (if any) is used before the first feedback vertex set vertex, after the last feedback vertex set vertex, and between each two feedback vertex set vertices. There are $2^\phi$ choices for which vertices of the feedback vertex set are present, at most $\phi!$ choices for how they are ordered, $\tbinom{n-\phi}{2}+(n-\phi)+1$ choices for either a path in the forest, a single vertex in the forest, or a direct connection between two feedback vertex set vertices, and at most $\phi+1$ paths in the forest (or direct connections) that must be chosen.
\end{proof}

\begin{corollary}
Path reconfiguration (in both the decision and optimization problems) can be solved by breadth-first search of the state space in time that is polynomial whenever the feedback vertex set number is bounded by a constant.
\end{corollary}

Because its exponent depends on $\phi$, the resulting parameterized algorithm does not belong to the complexity class $\mathsf{FPT}$ of fixed-parameter tractable algorithms, but rather to the class $\mathsf{XP}$ of algorithms that are polynomial for constant parameter values but with an exponent depending on the parameter.

\subsection{Other tree-like graph classes}

It is natural to consider whether path reconfiguration might be solved more easily for some other classes of graphs that have a tree-like structure but whose circuit rank and feedback vertex set number may be unbounded. These include, for instance, the cactus graphs (graphs for which each biconnected component is either a single edge or a simple cycle) and the block graphs (graphs for which each biconnected component is a complete subgraph). We leave as open for future research the questions of whether path reconfiguration can be solved efficiently in these graphs.

Solving the path reconfiguration problem for the block graphs would also necessarily involve solving it for the complete graphs. The path reconfiguration decision problem is trivial for complete graphs:
if the paths have fewer than $n$ vertices, the answer is always yes, and if they have exactly $n$ vertices then the answer is yes if and only if one path is a cyclic shift or the reversal of a cyclic shift of the other. However, even for complete graphs the path reconfiguration optimization problem appears to be nontrivial.

\else

In the full version of this paper we show that several special classes of graphs have polynomial algorithms for path reconfiguration regardless of path length. The prototypical example are the trees, for which the existence of a polynomial time algorithm follows immediately from the fact that any $n$-vertex tree has $O(n^2)$ distinct paths. In the full version, we refine this idea and provide a linear time algorithm for path reconfiguration in trees. 

We also observe that the graphs of bounded circuit rank, and the graphs of bounded feedback vertex number, have polynomial algorithms for path reconfiguration, because in these graphs
the size of the state space (the number of distinct paths in the graph) is bounded by a polynomial.
For circuit rank the exponent of the polynomial is a constant, and we obtain a fixed-parameter tractable algorithm. For feedback vertex number,
the exponent depends on the feedback vertex number.
We defer the details to the full version of the paper.

\section{Hardness}

\begin{figure}[htb]
\centering\includegraphics[scale=0.3]{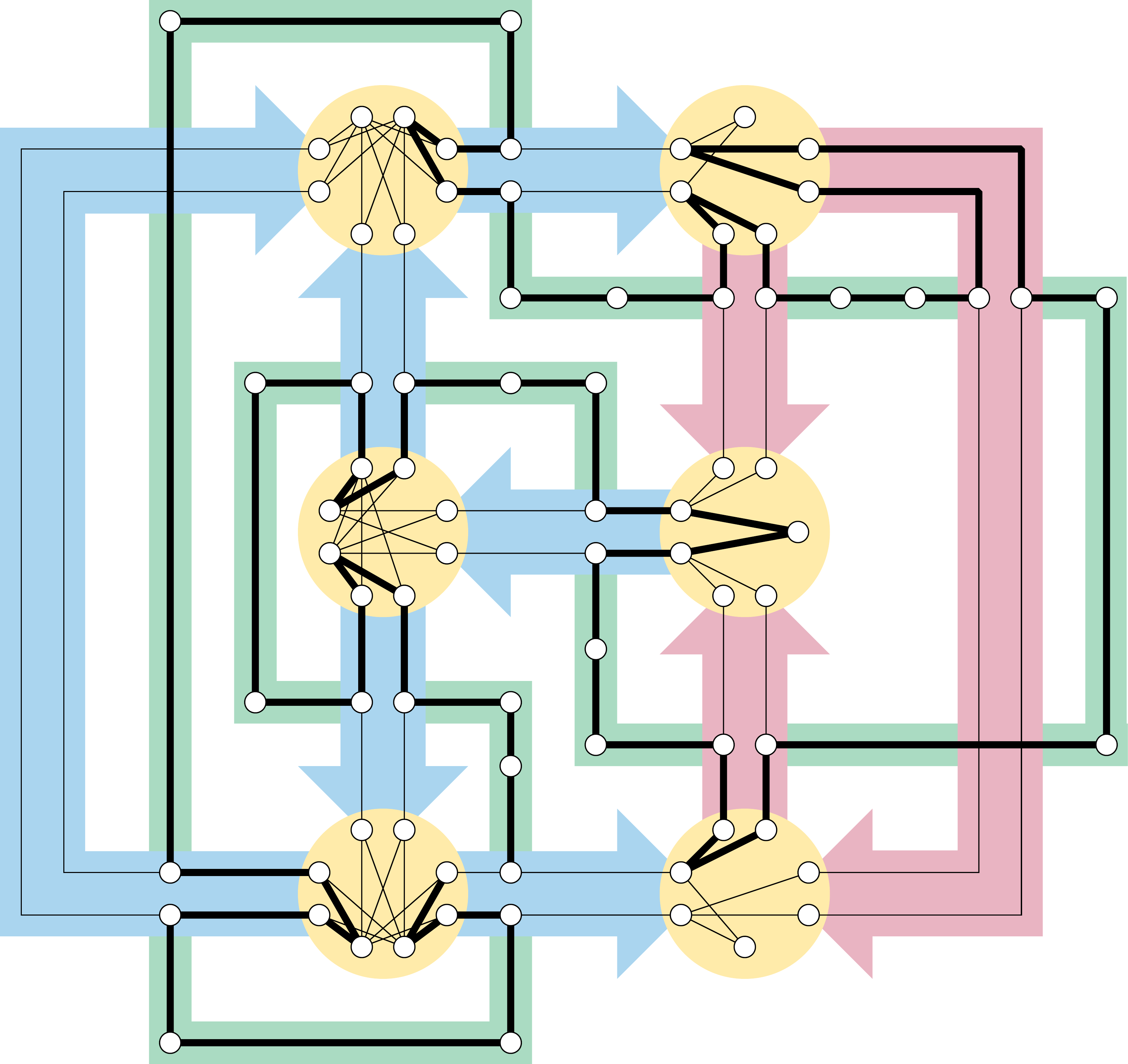}
\caption{Reduction from nondeterministic constraint logic to path reconfiguration. The underlying constraint logic instance has six vertices (yellow shaded circles) and nine edges (thick red and blue shaded arrows). Within each of the shaded circles is a vertex gadget of our reduction, and within each thick shaded arrow is an edge gadget of our reduction. The thin green shaded regions contain  connection gadgets of our reduction, which the path that is undergoing reconfiguration uses to pass from one edge or vertex gadget to another. The heavy black edges depict one possible state of the path to be reconfigured.}
\label{fig:reduction}
\end{figure}

In the full version of this  paper we describe a reduction from nondeterministic constraint logic showing that path reconfiguration (with unbounded path length) is $\mathsf{PSPACE}$-complete even on graphs of bounded bandwidth. This result rules out the possibility (unless $\mathsf{P}=\mathsf{PSPACE}$) that our results on tree-like graph classes from \autoref{sec:treelike} can be extended to another tree-like class of graphs, the graphs of bounded treewidth.

\begin{theorem}
The path reconfiguration decision problem is $\mathsf{PSPACE}$-complete, even for graphs of bounded bandwidth.
\end{theorem}

An example of our reduction is depicted in \autoref{fig:reduction}.

\bibliographystyle{splncs}
\bibliography{snake}
\end{document}
\fi

%%%%%%%%%%%%%%%%%%%%%%%%%%%%%%%%

\section{In trees}
\label{sec:trees}

In this section we provide a linear time algorithm for the path reconfiguration optimization problem in trees. 
The input is a tree $T$, a starting path $P$ and an ending path $Q$ that we wish to reconfigure $P$ to.
Our strategy is to identify intermediate positions $P'$, $Q'$ such that $P$ can be reconfigured to $P'$, $Q$ can be reconfigured to $Q'$, and $P'$, $Q'$ are in \emph{critical position} meaning that $P' \cup Q'$ is contained in a path in $T$---which makes it trivial to reconfigure $P'$ to $Q'$. 
See \autoref{fig:critical-disjoint} and \autoref{fig:critical-overlap}. 
We will show that the possible critical positions can be explored efficiently.
To reconfigure $P,Q$ to a critical position $P',Q'$ we will give a reduction to  
a subproblem that we call \emph{sub-path reconfiguration}.
We first show how to solve sub-path configuration and then how to reduce the general problem to subpath reconfiguration.

%%%%%%%%%%%%%%%%%%%%%%%%%%%%%%%%
\subsection{Sub-path reconfiguration}

\begin{definition}
\emph{Sub-path reconfiguration} is a problem in which we are given a tree $T$, a path $P$ in $T$, and a sub-path of $P$ from one given vertex $u$ to another given vertex $v$. 
We use the notation $(T,P,u,v)$ to denote this sub-path reconfiguration problem.
The object of the problem is to reconfigure $P$ (via the same steps as in the path reconfiguration problem) within $T$ to any other path $D$ that 
has an endpoint at $u$ and 
contains the same sub-path from $u$ to $v$, and such that each intermediate path of the reconfiguration process also contains the sub-path from $u$ to $v$. As with path reconfiguration we define two variants of this problem, the decision problem in which the goal is to 
decide if such a reconfiguration exists 
and the optimization problem in which the goal is to find the shortest possible reconfiguration.
\end{definition}

We will reduce sub-path reconfiguration to ``simpler'' instances of the same problem, where ``simpler'' means that the subpath becomes larger.

\begin{definition}
Let $(T,P,u,v)$ be a sub-path reconfiguration problem. We define a \emph{suitable detour vertex} to be a vertex $w$ in $P$ that is not internal to the subpath from $u$ to $v$ and
such that there exists a \emph{suitable detour path} of length $|P|$ in $T$ that starts at $u$ and coincides with $P$ from $u$ to $w$. We define the \emph{first suitable detour vertex} to be the suitable detour vertex farthest from $v$ (and therefore closest to the end of $P$), and we define a \emph{first suitable detour path} to be a suitable detour path for $w$. 
See \autoref{fig:subpath-reconfig}(b).   
\end{definition}

\paragraph{Reduction for subpath reconfiguration:}
\begin{itemize}
\item Find $w$, the first suitable detour vertex.  If $w$ does not exist then declare failure, and otherwise let $D$ be a first suitable detour path.
\item Let $x$ be the neighbour of $u$ in $P$ that is farthest from $v$.  Solve the subpath reconfiguration problem $(T,P, w, x)$.  
Let $\sigma_1$ be a shortest solution, and suppose it reconfigures $P$ to $P'$.  See \autoref{fig:subpath-reconfig}(b). 
\item Observe that $P'$ and $D$ overlap on the subpath $S$ from $u$ to $w$ and that their union is a simple path.  Let $\sigma_2$ be the reconfiguration sequence of length $|P| - |S|$ that slides $P'$ to $D$. 
\item Return the concatenation $\sigma_1 \sigma_2$.
\end{itemize}

We will show how to find $w$ and $D$ in the following subsection.  Note that in case there are multiple choices for the first suitable detour path $D$, all choices give the same length reconfiguration sequence.

\begin{figure}[t]
\centering\includegraphics[width=\textwidth]{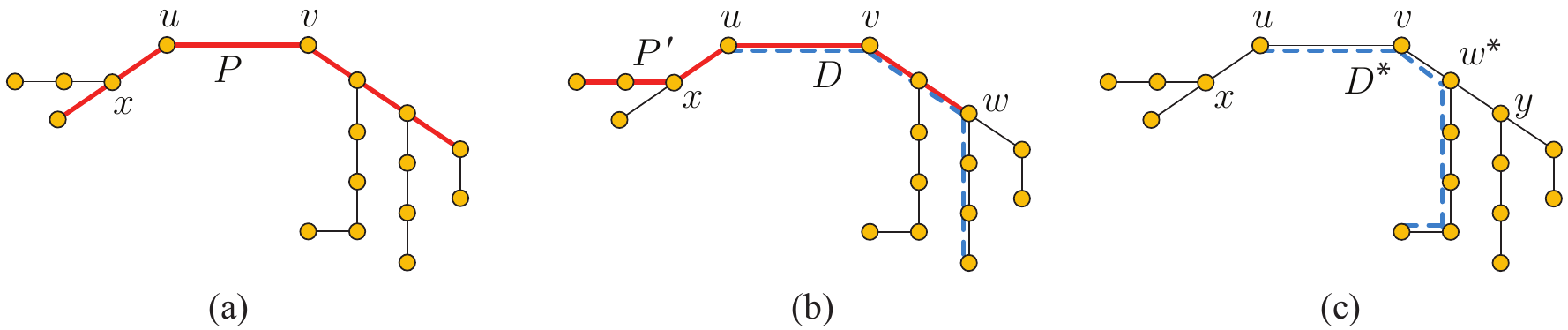}
\caption{Subpath reconfiguration: (a) the input; (b) the first suitable detour vertex $w$ and first suitable detour path $D$, and the path $P'$ used by the algorithm; (c) a path $D^*$ used by the optimum reconfiguration sequence.
}
\label{fig:subpath-reconfig}
\end{figure}

\begin{lemma}
\label{lem:sub-path-reduction}
The above algorithm finds a shortest reconfiguration sequence for the subpath reconfiguration problem $(T,P,u,v)$. 
\end{lemma}
\begin{proof}
Consider a shortest reconfiguration sequence $\sigma^*$ that solves $(T,P,u,v)$, and let $D^*$  be the final path at the end of $\sigma^*$.  Let $S^*$ be the maximal subpath shared by $P$ and $D^*$ and suppose that $S^*$ goes from $u$ to $w^*$.    Note that $S^*$ includes $v$.
Let $y$ be the neighbour of $w^*$ in $P$ that is outside $S^*$---note that $y$ exists otherwise $P=D^*$.
See \autoref{fig:subpath-reconfig}(c). 
Initially, $P$ contains both edges $u x$ and $w^* y$.  Consider the first time during $\sigma^*$ when one of these edges leaves $P$.  It cannot be $u x$, otherwise we would have solved the subpath reconfiguration before reaching $D^*$.  
Therefore there must be some shortest prefix $\sigma^*_1$ of $\sigma^*$ that reconfigures $P$ to a path $R^*$ that ends at $w^*$ and includes $S^* \cup u x$.  Furthermore, all edges of  $S^* \cup u x$ are included throughout $\sigma^*_1$.
Then $R^*$ and $D^*$ are in critical position sharing the subpath $S^*$.  $R^*$ can be reconfigured to $D^*$ by sliding it along $Q^*$ for $|P| - |S^*|$ steps, and there is no shorter reconfiguration because the paths lie in a tree.  Thus $|\sigma^*| = |\sigma^*_1| + |P| - |S^*|$.

We now compare what the algorithm does.     
Because $w^*$ is a suitable detour vertex for $(T,P,u,v)$, we know that a first suitable detour vertex, $w$, exists, and that $w$ is at least as far from $v$ as $w^*$.   As in the algorithm, let $S$ be the subpath from $u$ to $w$.  Then $w^*$ lies in $S$ and $|S| \ge |S^*|$.  The algorithm finds a shortest solution $\sigma_1$ to the subpath reconfiguration problem $(T,P,w,x)$.  We claim that $|\sigma_1| \le |\sigma^*_1|$.  This is because  $\sigma^*_1$ must move the endpoint of $P$ to $w$ before it can reach $w^*$.  
The solution returned by the algorithm has length $|\sigma_1| + |P| -|S| \le  |\sigma^*_1| + |P| - |S^*| = |\sigma^*|$.  Thus the algorithm finds a shortest reconfiguration sequence.   
\end{proof}

%%%%%%%%%%%%%%%%%%%%%%%%%%%%%%%%
\subsection{Sub-path reconfiguration algorithm}

To solve sub-path reconfiguration quickly, we need to be able to quickly identify the first suitable detour vertex for a given sub-path reconfiguration problem.
And we need to find a first suitable detour path.

We can assume without loss of generality that the initial path $P$ we are given consists of the vertices $v_0, v_1,\dots v_{|P|}$ in order.
We say that $v_0$ is the \emph{left} endpoint of $P$ and that $v_{|P|}$ is the \emph{right} endpoint.
Based on this orientation of $P$, we can define two quantities $\Delta_L(k)$ and $\Delta_R(k)$ at each vertex $v_k$, as follows:
\begin{itemize}
\item Let $\Delta_L(k)$ be $\ell-k$, where $\ell$ is the length of the longest path in $T$ that ends at $v_k$ and does not use vertex $v_{k+1}$. That is, $\Delta_L(k)$ measures how much deeper than $v_0\dots v_k$ a path from $v_k$ can go. If we root $T$ at $v_{|P|}$, then $\ell$ is just the height (maximum distance from a leaf) of $v_k$, so all values $\Delta_L(k)$ can be computed in linear time by a simple bottom-up calculation on this rooted tree.
\item Let $\Delta_R(k)$ be $r-(|P|-k)$, where $r$ is the length of the longest path in $T$ that ends at $v_k$ and does not use vertex $v_{k-1}$. That is, $\Delta_R(k)$ measures how much deeper than $v_k\dots v_{|P|}$ a path from $v_k$ can go. If we root $T$ at $v_0$, then $r$ is just the height of $v_k$, so all values $\Delta_R(k)$ can be computed in linear time by a simple bottom-up calculation on this rooted tree.
\end{itemize}

\begin{observation}
\label{obs:test-suitable}
For a subpath reconfiguration problem $(T,P,v_i,v_j)$ with $i>j$,
$v_k$ is a suitable detour vertex if and only if $k\le j$ and 
$\Delta_L(k)\ge |P|-i$.
For a subpath reconfiguration problem $(T,P,v_i,v_j)$ with $i<j$,
$v_k$ is a suitable detour vertex if and only if $k\ge j$ and $\Delta_R(k)\ge i$.
\end{observation}

\begin{lemma}
\label{lem:find-fsd}
For a subpath reconfiguration problem $(T,P,v_i,v_j)$,
we can find the first suitable detour vertex $v_k$ (if it exists) in time $O(1+|k-j|)$.
\end{lemma}

\begin{proof}
We test the sequence of vertices $v_j,v_{j+1},\dots$ (if $i<j$)
or the sequence $v_j,v_{j-1},\dots$ (if $i<j$),
using \autoref{obs:test-suitable} to check in constant time whether each vertex in the sequence is a suitable detour vertex, until either reaching the end of path $P$ or finding a vertex that is not a suitable detour. If $v_j$ is not itself a suitable detour vertex, then no suitable detour exists.
If we reach the end of the sequence, then the last vertex that was tested (the vertex at the end of the sequence) is the first suitable detour vertex. If we find a vertex other than $v_j$ that is not a suitable detour vertex, then the previous vertex in the sequence is the first suitable detour vertex.
\end{proof}

\begin{lemma}
\label{lem:find-detour-path}
After linear-time preprocessing we can find a suitable detour path from the first suitable detour vertex, in time proportional to its length.
\end{lemma}

\begin{proof}
To preprocess the input,
let $F$ be the forest formed from $T$ by removing the edges of $P$.
We root each tree of this forest at its unique vertex in $P$ (directing the edges of the forest away from $P$) and with a single bottom-up traversal of the tree we calculate for each vertex $u$ of $F$ two pieces of information: the length of the longest directed path starting at $u$, and the identity of a child of $u$ that can be the neighbor of $u$ on one such longest path.

When we find a suitable detour vertex $v_k$,
the length of the corresponding detour path should be $|P|-|i-k|$.
We can find a path of this length, in time proportional to its length,
by starting from $v_k$ and following a sequence of pointers to children on longest paths.
\end{proof}

\begin{lemma}
\label{lem:sub-paths-in-trees}
We can solve the sub-path reconfiguration optimization problem in trees in linear time.
\end{lemma}

\begin{proof}
We repeatedly apply \autoref{lem:find-fsd} and \autoref{lem:find-detour-path} to find the first suitable detour vertex and its detour path, and then \autoref{lem:sub-path-reduction} to reduce the problem to a sub-problem of the same type, until either reaching a sub-problem $(T,P,a,b)$ where $a$ is one of the endpoints of $P$ (in which case the solution is the empty reconfiguration sequence) or reaching a sub-problem with no suitable detours (in which case there is no reconfiguration sequence).

Each reduction includes at least one new vertex $x$ in the sub-path between the two vertices defining the reduced sub-problem, so the total number of levels of reduction is at most $|P|$.
The time for each reduction step is proportional to a constant, plus the length of the sub-path from $v_j$ to $v_k$ (according to the notation of \autoref{lem:find-fsd}), plus the length of the detour path. The sub-paths are all edge-disjoint, as each one is outside the path between the two endpoints of the sub-path reconfiguration problem before the reduction, and inside this path after the reduction. Therefore their total length, and the total time for finding the whole sequence of reductions, adds to $|P|$.

The detour paths are all disjoint, so their total length, and the time for finding them, adds to $O(|T|)$. The subsequences $\sigma_2$ concatenated in each reduction step each involve sliding along one of these detour paths, and each edge of such a path contributes two units to the total length of these subsequences (one unit for the step in which the path first enters that edge, and a second unit for the step in which it leaves the edge again).
Therefore, the total length of all these subsequences is $O(|T|)$ and even if we write them out explicitly (as a single piece of information per step, the identity of the edge that the path slides into) the total time for doing so is $O(|T|$).
\end{proof}

%%%%%%%%%%%%%%%%%%%%%%%%%%%%%%%%
\subsection{Reducing path reconfiguration to sub-path reconfiguration}

Consider a path reconfiguration problem $(T,P,Q)$ for starting path $P$ and ending path $Q$ in tree $T$.
We will separate into two cases depending on whether $P$ and $Q$ are edge-disjoint or not.

\subsubsection{Edge-disjoint paths} 
We begin with an algorithm 
to reduce path reconfiguration to sub-path reconfiguration in the case where the $P$ and $Q$ are edge-disjoint, as illustrated in \autoref{fig:critical-disjoint}(a).  
Let $R$ be the (unique) path in $T$ from $P$ to $Q$, connecting vertices $p\in P$ and $q\in Q$.   It is possible that $p=q$ and $R$ has no edges.
The idea of our solution is to find a critical position $P'$, $Q'$ where $P'$ ends at $p$ and $Q'$ ends at $q$.

\begin{figure}[t]
\centering\includegraphics[scale=0.8]{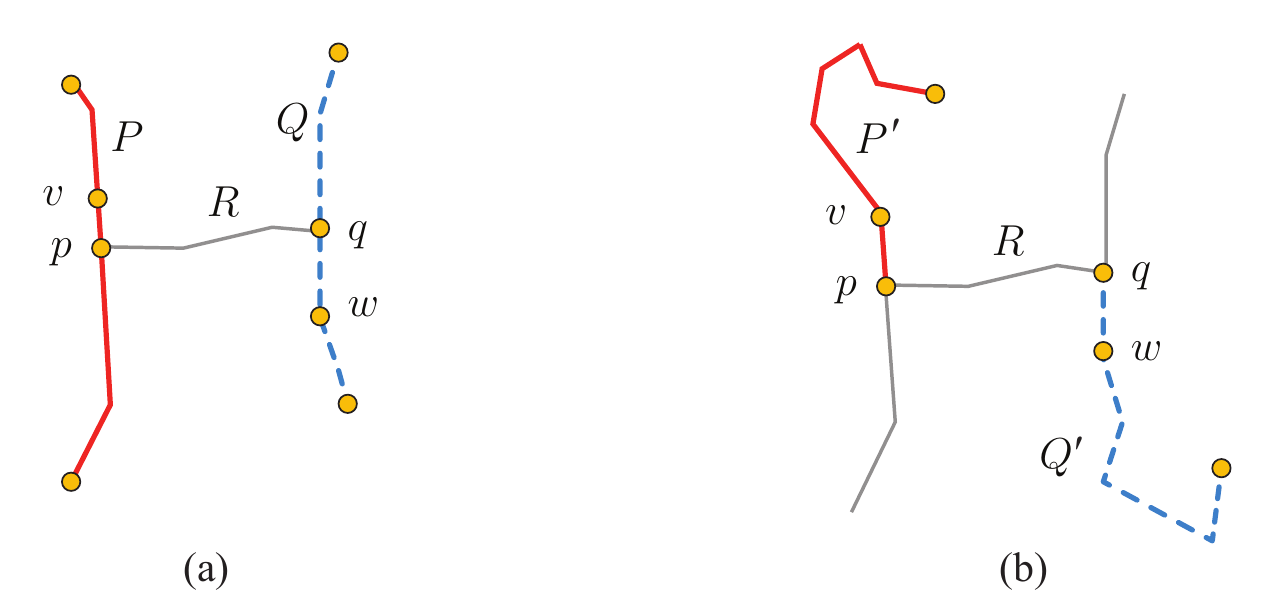}
\caption{Reconfiguring path $P$ to path $Q$ in the case where $P$ and $Q$ are edge-disjoint: (a) the initial situation; (b) one of the four possible critical positions (each of $P$/$Q$ can go up/down).}
\label{fig:critical-disjoint}
\end{figure}

\begin{figure}[t]
\centering\includegraphics[scale=0.45]{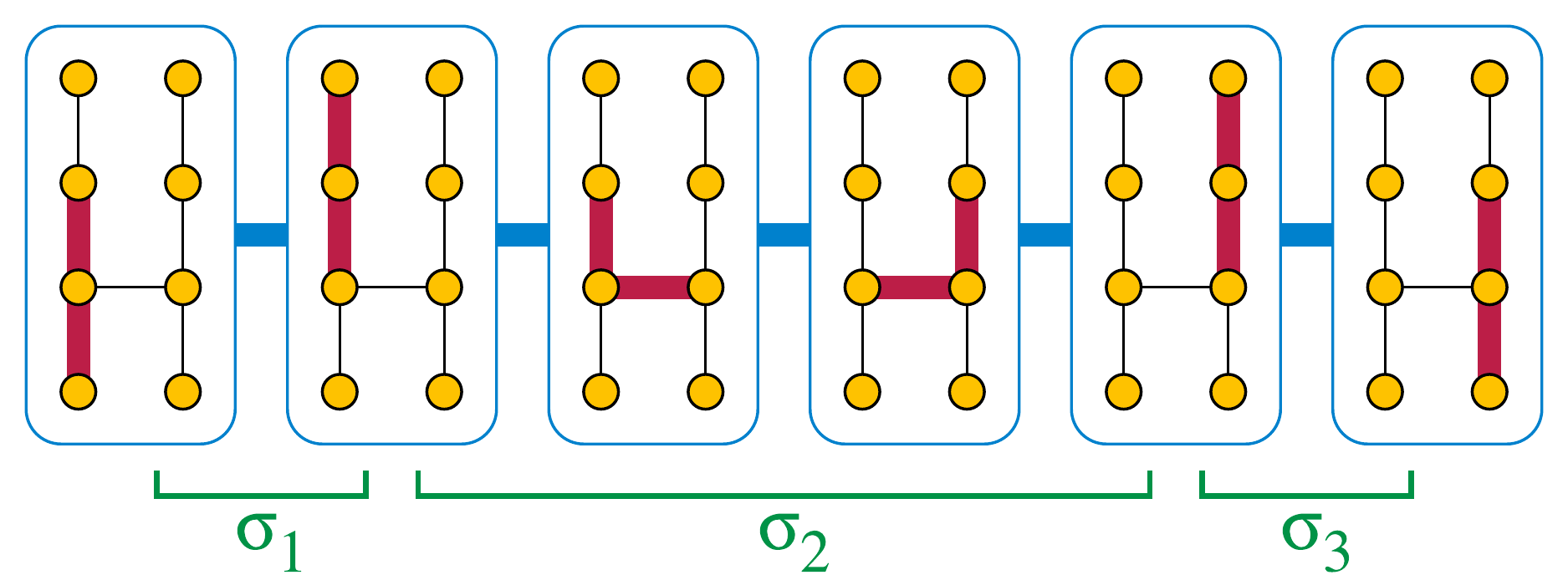}
\caption{A specific example of the algorithm reconfiguring one path to another when they are edge-disjoint.
}
\label{fig:tree-disjoint-case}
\end{figure}

\paragraph{Reduction for edge-disjoint paths:}
\begin{itemize}
\item For each $v$ a neighbor of $p$ in $P$, solve the sub-path reconfiguration optimization problem $(T,P,p,v)$.  If there are no solutions, then declare failure, and otherwise, let $\sigma_1$ be a shortest solution, and suppose it reconfigures $P$ to $P'$.
\item For each $w$ a neighbor of $q$ in $Q$, solve the sub-path reconfiguration optimization problem $(T,Q,q,w)$.  If there is no solution, then declare failure, and otherwise, let $\sigma_3$ be the time-reversal of a shortest solution, and suppose it reconfigures $Q'$ to $Q$.
\item Observe that $P'\cup Q'\cup R$ is a simple path, 
so the paths $P'$, $Q'$ are in critical position.  
Let $\sigma_2$ be the reconfiguration sequence of length $|P| + |R|$ that slides $P'$ along the path $P'\cup Q'\cup R$ to $Q'$.  
\item Return the concatenation $\sigma_1 \sigma_2 \sigma_3$.
\end{itemize}
This construction is illustrated in general in \autoref{fig:critical-disjoint} and for a specific example in \autoref{fig:tree-disjoint-case}.

\begin{lemma}
\label{lem:tree-disjoint}
The above algorithm finds a 
shortest reconfiguration sequence for $(T,P,Q)$ when $P$ and $Q$ are edge disjoint.
Furthermore, the algorithm solves at most four subpath reconfiguration problems. 
\end{lemma}
\begin{proof}
We first show that any shortest reconfiguration sequence, $\sigma^*$,  from $P$ to $Q$ necessarily has an intermediate critical position where $P'$ ends at $p$ and $Q'$ ends at $q$.  
If $p$ is an endpoint of $P$ then $P$ is already in critical position.
Otherwise, 
let the two neighbors of $p$ in $P$ be $v'$ and $v''$. 
The original path $P$ uses both edges $pv'$ and $pv''$, and this remains true during $\sigma^*$ until the 
reconfigured path ends at $p$.  Since $Q$ uses neither of these edges, 
there must be some shortest prefix $\sigma^*_1$ of $\sigma^*$ that reconfigures $P$ to a path, $P^*$ that ends at $p$ and uses exactly one of these two edges.  (Note:  $\sigma^*_1$ is empty if the initial $P$ ends at $p$.)
By an argument that is symmetric under time-reversal symmetry, there must be some shortest suffix $\sigma^*_3$ of $\sigma^*$ that reconfigures a path, $Q^*$ to $Q$, where $Q^*$  ends at $q$ and uses exactly one of its two incident edges in $Q$. 
Observe that $P^*$ and $Q^*$ are edge-disjoint and that neither shares an edge with $R$.
Furthermore, 
the prefix $\sigma^*_1$ and suffix $\sigma^*_3$ are disjoint, because any step that is intermediate to both of them would correspond to a path in $T$ that uses both edges of $P$ incident to $p$ and both edges of $Q$ incident to $q$, but no such path can exist.  Thus $P^*$ and $Q^*$ form a critical position.

$R^*$ can be reconfigured to $D^*$ by sliding it along $Q^*$ for $|P| - |S^*|$ steps, and there is no shorter reconfiguration because the paths lie in a tree.

We now examine how $\sigma^*$ solves the critical position.
Let $\sigma^*_2$ be the subsequence of $\sigma^*$ that reconfigures $P^*$ to $Q^*$, 
so that $\sigma^*$ is the concatenation $\sigma^*_1\sigma^*_2\sigma^*_3$.   
Note that $P^*\cup Q^*\cup R$ is a simple path of length $2|P|+|R|$.  
By sliding $P^*$ along this path, we can reconfigure $P^*$ to $Q^*$ in $|P| + |R|$ steps. 
Furthermore, because the graph is a tree, there is no shorter reconfiguration sequence from $P^*$ to $Q^*$.
Thus $|\sigma^*_2|=|P|+|R|$.

We now compare what the algorithm does.  Suppose that $P^*$ uses the edge $pv$ and $Q^*$ uses the edge $qw$.  The algorithm tries vertex $v$ and finds an optimum solution to the subpath reconfiguration problem $(T,P,p,v)$.  Since $\sigma^*_1$ is one possible solution, thus $|\sigma_1| \le |\sigma^*_1|$.
Similarly, the algorithm tries vertex $w$ and finds an optimum solution to the subpath reconfiguration problem $(T,Q,q,w)$. Since the time reversal of $\sigma^*_3$ is one possible solution, thus $|\sigma_3| \le |\sigma^*_3|$.  Furthermore, $\sigma_2$ and $\sigma^*_2$ have the same length $|P|+|R|$.  Therefore the algorithm finds a shortest reconfiguration sequence.
\end{proof}

\subsubsection{Overlapping paths} 

We now consider path reconfiguration where the paths $P$ and $Q$ share 
at least one edge. In particular, suppose they share 
a maximal subpath $S$ from vertex $u$ to vertex $v$. 
We will show that any shortest reconfiguration sequence involves a critical position $P',Q'$ of one of the following forms (see \autoref{fig:critical-overlap}):
\begin{enumerate}
\item $P'$ and $Q'$ are edge-disjoint,  and they both end at $u$ or they both end at $v$. 
\item $P'$ and $Q'$ share a non-empty prefix or suffix of $S$.  
\end{enumerate}
We will try all these types of critical positions to see which gives a shortest solution. 

Notation: Let $u_P$ be the neighbour of $u$ in $P$ but not in $Q$, if it exists, and let $u_Q$ be the neighbour of $u$ in $Q$ but not in $P$, if it exists.  Similarly, let $v_P$ be the neighbour of $v$ in $P$ but not in $Q$ and let $v_Q$ be the neighbour of $v$ in $Q$ but not in $P$ if they exist.  Let $u'$ be the neighbour of $u$ in $S$ and let $v'$ be the neighbour of $v$ in $S$.  
See \autoref{fig:critical-overlap}(a).

\begin{figure}[t]
\centering\includegraphics[scale=0.8]{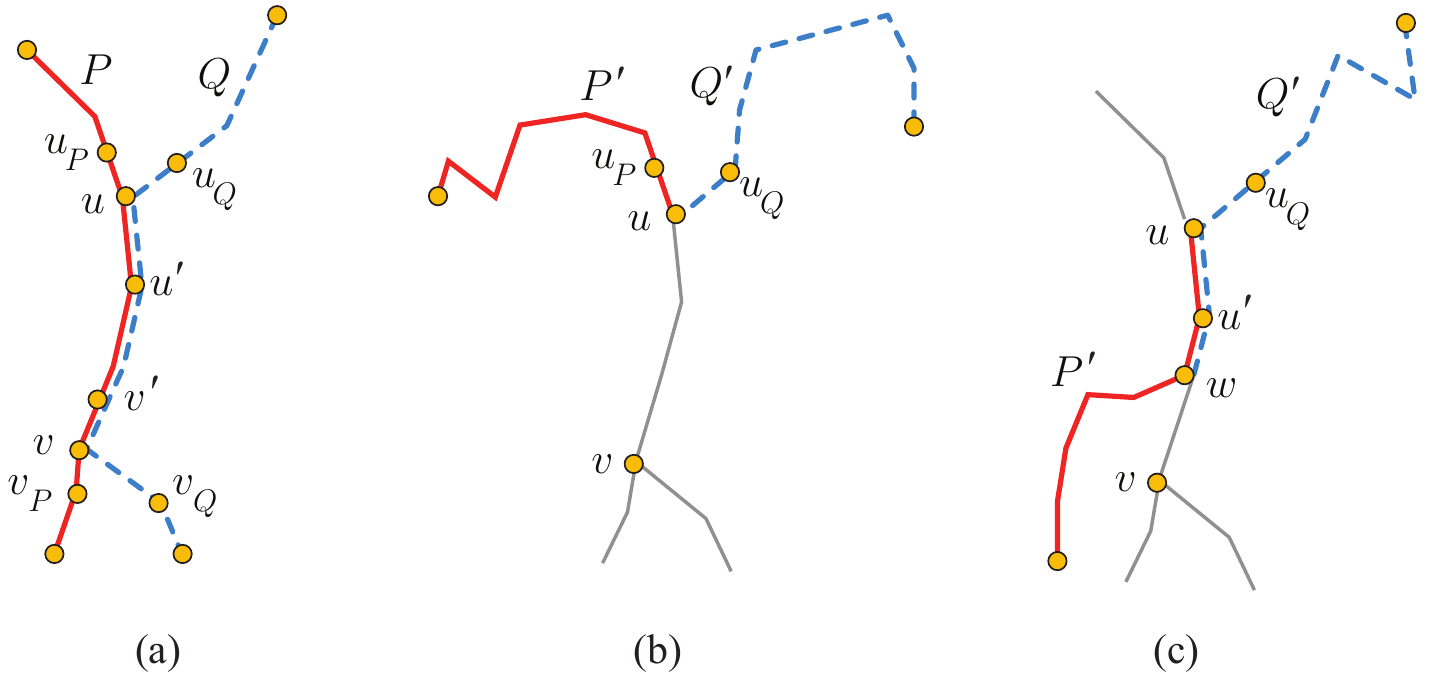}
\caption{Reconfiguring path $P$ to path $Q$ in the case where they overlap from vertex $u$ to vertex $v$: (a) the initial situation; (b) a critical position where $P'$ and $Q'$ are edge-disjoint and both end at $u$ (option 1); 
(c) a critical position where $P'$ and $Q'$ share the subpath from $u$ to $w$ (option 3).
}
\label{fig:critical-overlap}
\end{figure}

\paragraph{Reduction for paths that share edges:}
Try all four of the following options and return one that gives the shortest reconfiguration sequence:
\begin{enumerate}
\item  $P'$, $Q'$ share only vertex $u$.   See \autoref{fig:critical-overlap}(b).
\begin{itemize}
	\item Solve the sub-path reconfiguration problem $(T,P,u, u_P)$. If there is no solution or $u_P$ does not exist, then abandon this case, and otherwise, let $\sigma_1$ be a shortest solution, and suppose it reconfigures $P$ to $P'$. 
	\item Solve the sub-path reconfiguration problem $(T,Q,u, u_Q)$. If there is no solution of $u_Q$ does not exist, then abandon this case, and otherwise, let $\sigma_3$ be the time-reversal of a shortest solution, and suppose it reconfigures $Q'$ to $Q$.
	\item Let $\sigma_2$ be the reconfiguration sequence of length $|P|$ that slides $P'$ to $Q'$.
	\item Add $\sigma_1 \sigma_2 \sigma_3$ to the set of candidate solutions.
\end{itemize}

\item $P'$, $Q'$ share only vertex $v$.  This is symmetric to option 1 and is solved similarly.

\item $P'$, $Q'$ share a prefix of $S$.  See \autoref{fig:critical-overlap}(c).
\begin{itemize}
	\item Solve the sub-path reconfiguration problem $(T,P,u, u')$. If there is no solution, then abandon this case, and otherwise, let $w$ be the first suitable detour, let $\sigma_1$ be a shortest solution, and suppose it reconfigures $P$ to $P'$.  
	\item If $w$ is not in the subpath $S$ then replace $w$ by $v$.
	Solve the sub-path reconfiguration problem $(T,Q,w, u_Q)$. If there is no solution, then abandon this case, and otherwise, let $\sigma_3$ be the time-reversal of a shortest solution, and suppose it reconfigures $Q'$ to $Q$.
	\item Let $S'$ be the prefix of $S$ from $u$ to $w$.  Note that $P'$ and $Q'$ are in critical position, sharing $S'$.  Let $\sigma_2$ be the reconfiguration sequence of length $|P| - |S'|$ that slides $P'$ to $Q'$.
	\item Add $\sigma_1 \sigma_2 \sigma_3$ to the set of candidate solutions.
\end{itemize}

\item $P'$, $Q'$ share a suffix of $S$. This is symmetric to option 3 and is solved similarly:

\begin{itemize}
	\item Solve the sub-path reconfiguration problem $(T,P,v, v')$. If there is no solution, then abandon this case, and otherwise, let $w$ be the first suitable detour, let $\sigma_1$ be a shortest solution, and suppose it reconfigures $P$ to $P'$.  
	\item If $w$ is not in the subpath $S$ then replace $w$ by $u$.
	Solve the sub-path reconfiguration problem $(T,Q,w, v_Q)$. If there is no solution, then abandon this case, and otherwise, let $\sigma_3$ be the time-reversal of a shortest solution, and suppose it reconfigures $Q'$ to $Q$.
	\item Let $S'$ be the suffix of $S$ from $w$ to $v$.  Note that $P'$ and $Q'$ are in critical position, sharing $S'$.  Let $\sigma_2$ be the reconfiguration sequence of length $|P| - |S'|$ that slides $P'$ to $Q'$.
	\item Add $\sigma_1 \sigma_2 \sigma_3$ to the set of candidate solutions.
\end{itemize}
 
\end{enumerate}
Specific examples are illustrated in \autoref{fig:tree-flip-case}, \autoref{fig:tree-zigzag-case}, and \autoref{fig:reconfig-detour-example}.

\begin{figure}[t]
\centering\includegraphics[scale=0.45]{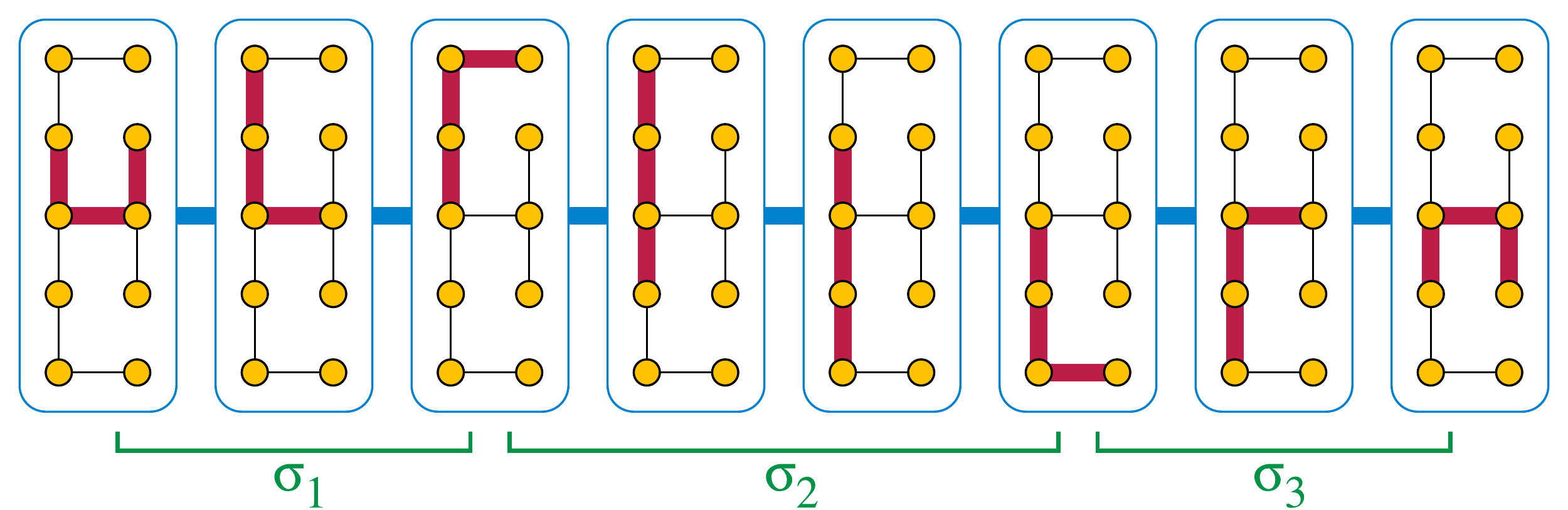}
\caption{The reconfiguration sequence found by the algorithm for two overlapping paths whose critical position involves edge-disjoint paths. 
}
\label{fig:tree-flip-case}
\end{figure}

\begin{figure}[t]
\centering\includegraphics[scale=0.45]{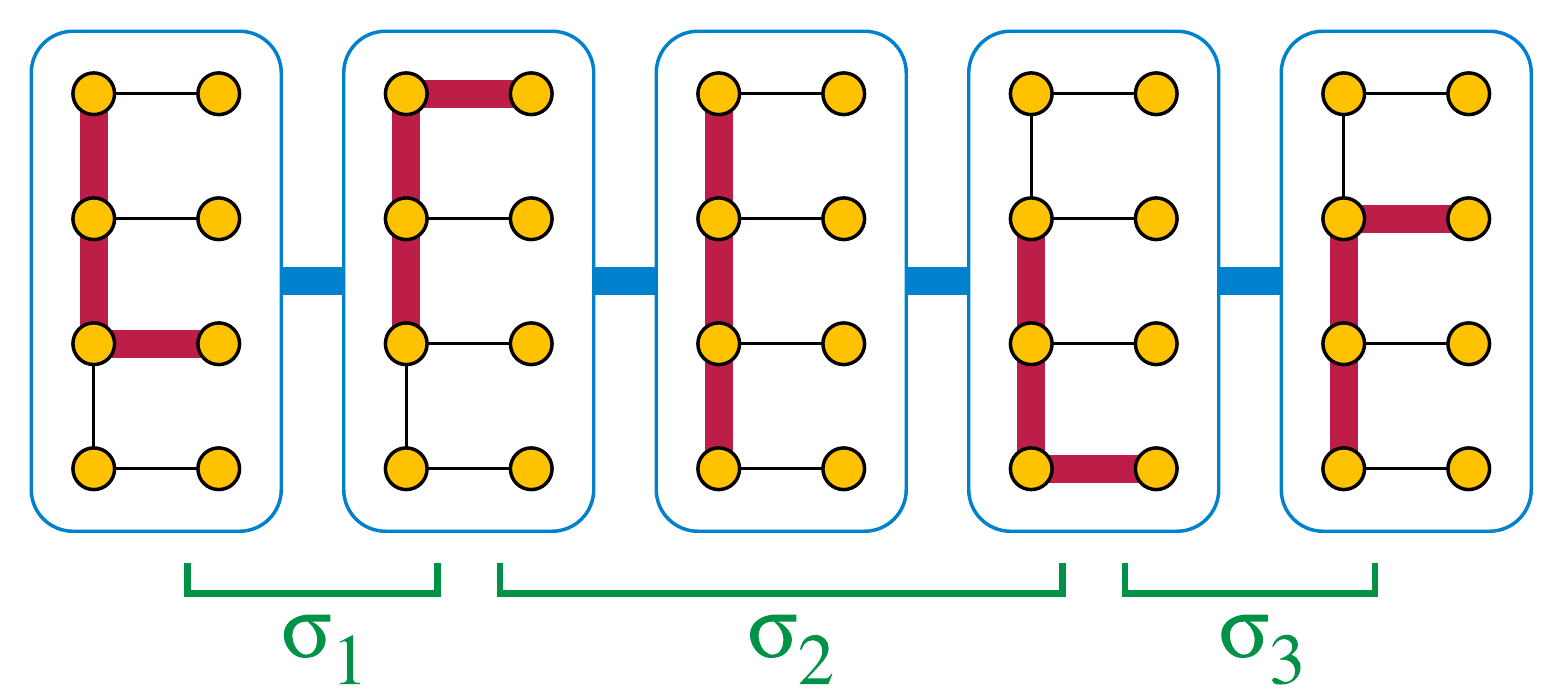}
\caption{The reconfiguration sequence found by the algorithm for two paths $P$ and $Q$ that overlap  and whose critical position overlaps on the same subpath.
}
\label{fig:tree-zigzag-case}
\end{figure}

\begin{figure}[t]
\centering\includegraphics[width=\textwidth]{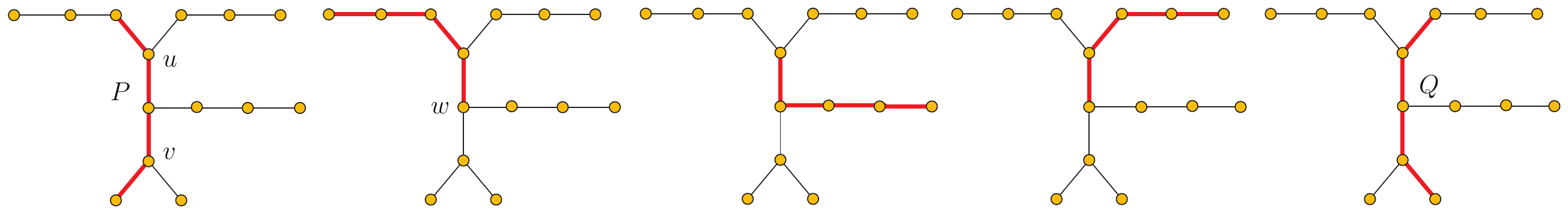}
\caption{
The reconfiguration sequence found by the algorithm for two paths that overlap on the subpath $S$ from $u$ to $v$  and whose critical position overlaps on a prefix of $S$. Not all steps are shown.}
\label{fig:reconfig-detour-example}
\end{figure}

\begin{lemma}
\label{lem:tree-overlap}
The above algorithm finds a 
shortest reconfiguration sequence for $(T,P,Q)$ when $P$ and $Q$ share at least one edge.
Furthermore, the algorithm solves at most eight subpath reconfiguration problems. 
\end{lemma}
\begin{proof}
The algorithm tries four options, each solving at most two subpath reconfiguration problems, which proves the ``furthermore'' statement.

We will show that any shortest reconfiguration sequence, $\sigma^*$, from $P$ to $Q$, necessarily has an intermediate critical position of one of the above types, and then argue that the algorithm finds a shortest solution.

We use notation as defined at the beginning of the algorithm.
Since $Q$ uses neither of the edges $u u_P$ nor $v v_P$, there must be some shortest prefix $\sigma^*_1$ of $\sigma^*$ that reconfigures $P$ to a path $P^*_1$ that does not use both these edges.  (It is possible that $P^*_1 = P$.)
It suffices to consider the case where $P^*_1$ does not use $v v_P$---the other case is obtained by exchanging the roles of $u$ and $v$.  (And note that the algorithm treats $u$ and $v$ equally.)
Then $u_P$ must exist, and $P^*_1$ must have one endpoint at $v$ and must include all the edges of $S \cup u u_P$.  Furthermore,  all the edges of $S \cup u u_P$ are included throughout $\sigma^*_1$.  See the first frames of \autoref{fig:reconfig-detour}.

We now apply the same argument to the time-reversal of the reconfiguration sequence $\sigma^*$.
Let $u_Q$ be the neighbour of $u$ in $Q$ but not in $P$, and let $v_Q$ be the neighbour of $v$ in $Q$ but not in $P$, if these vertices exist. 
Then there must be some shortest suffix $\sigma^*_3$ of $\sigma^*$ that reconfigures a path, $Q^*_1$ to $Q$, where either: $Q^*_1$  ends at $u$ and includes all the edges of $S \cup v v_Q$; or $Q^*_1$ ends at $v$ and includes all the edges of $S \cup u u_Q$.   In both cases, the included edges are in the path throughout $\sigma^*_3$.
We deal with these cases separately.

\noindent{\bf Case 1.}  $Q^*_1$  ends at $u$ and includes all the edges of $S \cup v v_Q$.
Then $P^*_1$ and $Q^*_1$ are in critical position and overlap on $S$.  We claim that the subsequence, $\sigma^*_2$ of $\sigma^*$ that reconfigures $P^*_1$ to $Q^*_1$ does so in $|P| - |S|$ steps by sliding $P_1$ along to $Q_1$.  
As before, the argument is that sliding is feasible, and there is no shorter solution because the graph is a tree.

Next we show that, in this case, the algorithm finds a reconfiguration sequence of length $|\sigma^*|$.
In option 4, the algorithm solves the subpath reconfiguration problem $(T,P,v,v')$.  Since $\sigma^*_1$ is a possible solution, the algorithm finds a reconfiguration sequence $\sigma_1$ of length at most $|\sigma^*_1|$.  Furthermore, $P^*_1$ includes all the edges of $S$, so the first suitable detour vertex $w$ for $(T,P,v,v')$ lies outside $S$.  Thus, the algorithm next resets $w$ to $u$ and  solves the subpath reconfiguration problem $(T,Q,u,v_Q)$.  The time-reversal of $\sigma^*_3$ is a possible solution, so the algorithm finds a reconfiguration sequence $\sigma_3$ of length at most $|\sigma^*_3|$.   Finally, the algorithm finds $\sigma_2$, a slide of length $|P| - |S| = |\sigma^*_2|$.  In total the algorithm finds a reconfiguration sequence of length at most  $|\sigma^*|$.  Since $\sigma^*$ was shortest, the algorithm finds a minimum length sequence.    

\noindent{\bf Case 2.}  $Q^*_1$ ends at $v$ and includes all the edges of $S \cup u u_Q$.
See \autoref{fig:reconfig-detour}.  
Consider what happens in $\sigma^*$ after $P^*_1$.  So long as the path uses both $u u'$ and $u u _P$, it cannot be reconfigured to $Q^*_1$.  Therefore, there must be some first time after $P^*_1$ where the path gets reconfigured to a path $P^*_2$ that has an endpoint at $u$ and uses exactly one of the two edges $u u'$ or $u u_P$.  
Applying the same argument to the time-reversed sequence, there must be some last time before $Q^*_1$ where the path $Q^*_2$ has an endpoint at $u$ and uses exactly one of the two edges $u u'$ or $u u_Q$.  Observe that $P^*_2$ does not occur after $Q^*_2$ in $\sigma^*$, although they may be equal.  
We consider 4 subcases.

\noindent{\bf Case (i).}
$P^*_2$ uses the edge $u u_P$ and $Q^*_2$ uses the edge $u u_Q$.   Then $P^*_2$ and $Q^*_2$ are in critical position, edge-disjoint, and sharing vertex $u$.   The algorithm finds this solution in option 1.  (The details of correctness are as in the proof of \autoref{lem:tree-disjoint}.)

\noindent{\bf Case (ii).}  $P^*_2$ uses the edge $u u_P$ and $Q^*_2$ uses the edge $u u'$.  Then in the step before $P^*_2$ the path---let's call it $R$---ends at $u'$ and uses $u_P u$ and $u u'$ so $R$ and $Q^*_2$ are in critical position (overlapping on $u u'$) so we could shortcut $\sigma^*$ by reconfiguring $R$ to $Q^*_2$, a contradiction to $\sigma^*$ being shortest.

\noindent{\bf Case (iii).} $P^*_2$ uses the edge $u u'$ and $Q^*_2$ uses the edge $u u_Q$.  This is ruled out by symmetry with the previous case. 

\begin{figure}[t]
\centering\includegraphics[width=\textwidth]{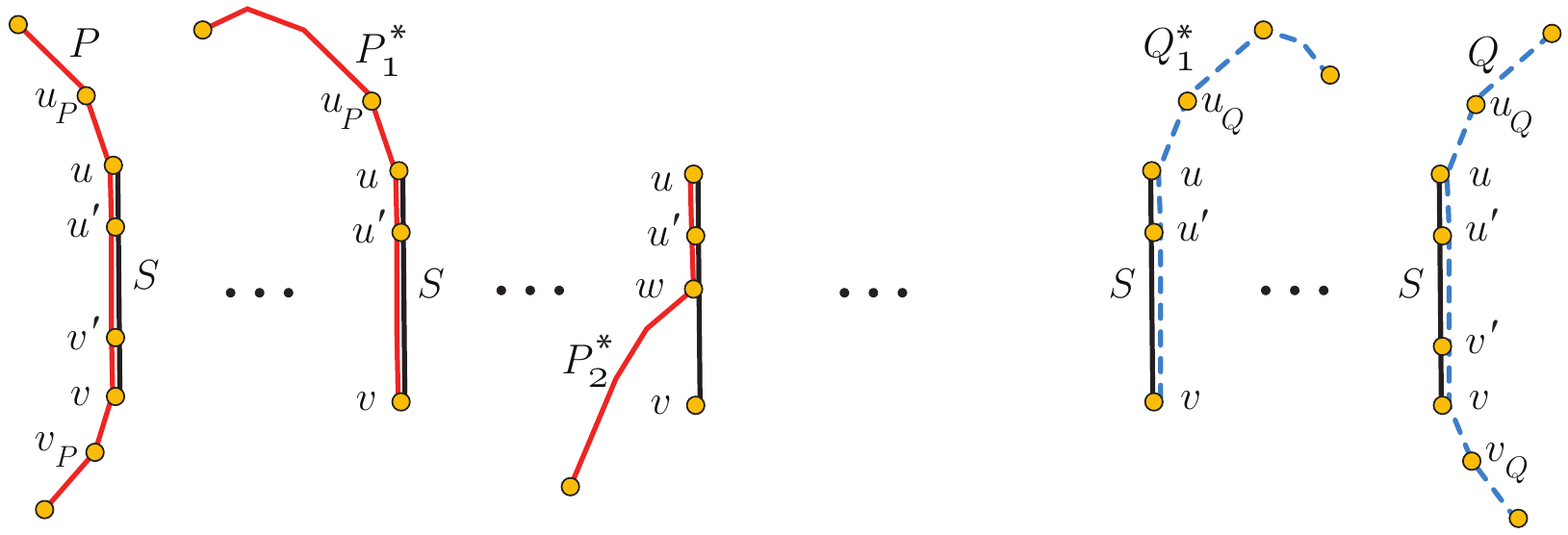}
\caption{The structure of an optimum reconfiguration sequence in Case 2 of \autoref{lem:tree-overlap}.
}
\label{fig:reconfig-detour}
\end{figure}

\noindent{\bf Case (iv).}
$P^*_2$ and $Q^*_2$ both  use the edge $u u'$.  Let $\sigma^*_A$ be the prefix of $\sigma^*$ that reconfigures $P$ to $P^*_2$ and let $\sigma^*_B$ be the suffix of $\sigma^*$ that reconfigures $Q^*_2$ to $Q$.  

In option 3, the algorithm solves the subpath reconfiguration problem $(T,P,u,u')$, 
with first suitable detour vertex $w$, reconfiguration sequence $\sigma_1$, and first suitable detour path $P'$.
Since $\sigma^*_A$ is a possible solution to $(T,P,u,u')$, the algorithm finds a reconfiguration sequence $\sigma_1$ of length at most $|\sigma^*_A|$.   

Consider an alternate algorithm that next solves the subpath reconfiguration problem $(T,Q,u,u')$.  Let  $\sigma'_2$ be the time-reversal of a shortest solution.  We first claim that $\sigma_1 \sigma'_2$ reconfigures $P$ to $Q$.   This is because the first suitable detour vertex $w$ is defined the same way for both $(T,P,u,u')$ and $(T,Q,u,u')$, and furthermore, we can use the same first suitable detour path $P'$
 for $(T,Q,u,u')$ (as noted after the subpath reconfiguration reduction).    

We next show that $\sigma_1 \sigma'_2$ is a shortest reconfiguration sequence from $P$ to $Q$.  We already noted that $|\sigma_1| \le |\sigma^*_A|$.  Similarly, the time-reversal of $\sigma^*_B$ is a possible solution to $(T,Q,u,u')$ so $|\sigma'_2| \le |\sigma^*_B|$.  
Thus $|\sigma_1 \sigma'_2| \le |\sigma^*_A| + |\sigma^*_B| \le |\sigma^*|$. 
Since $\sigma^*$ is a shortest reconfiguration sequence, therefore the alternate algorithm finds an optimum solution.  

Finally, we will show that the alternate algorithm and option 3 of the actual algorithm give the same length reconfiguration sequence.  In the alternate algorithm, $\sigma'_2$ is the time-reversal of a shortest solution to  $(T,Q,u,u')$.  
By \autoref{lem:sub-path-reduction}  $\sigma'_2 = \alpha \beta$ where $\beta$ is the time-reversal of a shortest solution to $(T,Q,w,u_Q)$ and $\alpha$ slides $P'$ to the path resulting from $\beta$ (these two paths overlap from $w$ to $u$). 
Option 3 of the algorithm finds a solution $\sigma_1 \sigma_2 \sigma_3$ where $\sigma_3$ is the time reversal of a shortest solution to $(T,Q,w, u_Q)$---which is the same subproblem solved in the alternate algorithm---and $\sigma_2$ slides $P'$ to the resulting path, which overlaps $P'$ from $w$ to $u$.
Therefore $|\beta| = |\sigma_3|$ and $|\alpha| =  |\sigma_2|$, which implies that  $|\sigma_1 \sigma'_2| = |\sigma_1 \sigma_2 \sigma_3|$.
Thus the alternate algorithm and the actual algorithm produce reconfiguration sequences of the same length.  This completes the proof that the algorithm finds a shortest reconfiguration sequence from $P$ to $Q$.
\end{proof}

\subsection{Overall algorithm}

Summarizing the results of this section, we have:

\begin{theorem}
We can solve the path reconfiguration optimization problem in trees in linear time.
\end{theorem}

\begin{proof}
We apply \autoref{lem:tree-disjoint} for edge-disjoint paths and \autoref{lem:tree-overlap} for overlapping paths to reduce the problem to the solution of a constant number of sub-path reconfiguration optimization problems, and \autoref{lem:sub-paths-in-trees} to solve each of these sub-path reconfiguration problems in linear time.
\end{proof}

\section{Hardness}

In this section we describe a reduction from nondeterministic constraint logic showing that path reconfiguration (with unbounded path length) is $\mathsf{PSPACE}$-complete even on graphs of bounded bandwidth. This result rules out the possibility (unless $\mathsf{P}=\mathsf{PSPACE}$) that our results on tree-like graph classes from \autoref{sec:treelike} can be extended to another tree-like class of graphs, the graphs of bounded treewidth. An example of our reduction is depicted in \autoref{fig:reduction}.

\begin{figure}[htb]
\centering\includegraphics[scale=0.35]{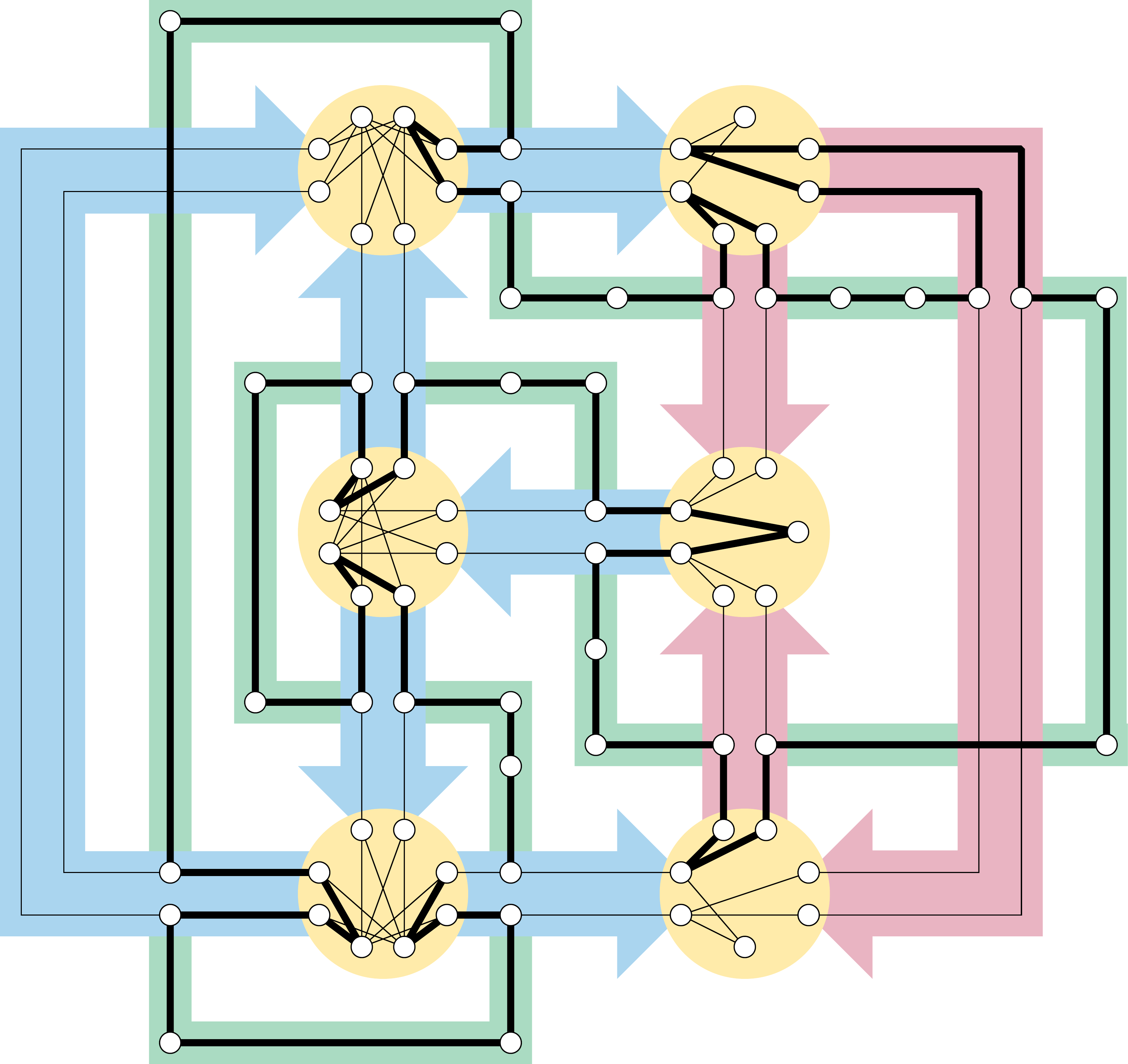}
\caption{Reduction from nondeterministic constraint logic to path reconfiguration. The shaded regions depict edge gadgets (blue and red), vertex gadgets (yellow), and connection gadgets (green). A valid path starting and ending near the lower left vertex gadget is shown by the heavy black edges.}
\label{fig:reduction}
\end{figure}

\subsection{Nondeterministic constraint logic}

For our hardness results, we will use \emph{nondeterministic constraint logic}, a reconfiguration problem on orientations of weighted undirected cubic graphs introduced by Demaine and Hearne~\cite{HeaDem-ICALP-02,HeaDem-09} as a fundamental hard problem for proving the hardness of a wide variety of games and puzzles.

An instance of nondeterministic constraint logic is determined by an undirected graph $G$ with exactly three edges per vertex. The edges have weight either one or two, and the total weight at each vertex is required to be even. A \emph{valid orientation} of this graph is an assignment of a direction to each edge so that the total incoming weight at each vertex is at least two. For the type of reconfiguration problem that we study here, we are interested in whether it is possible to get from a given starting orientation to a given goal orientation (both valid) by moves that reverse the direction of a single edge and preserve the validity of the orientation. This problem is known to be $\mathsf{PSPACE}$-complete, even when restricted to graphs of bounded bandwidth~\cite{vdZ-IPEC-15}.

By convention, when visualizing instances of nondeterministic constraint logic, the edges of weight one are drawn as red and the edges of weight two are drawn as blue. There are two types of vertex: \emph{AND gates} with two red edges and one blue edge, and \emph{OR gates} with three blue edges. For an AND gate, the blue edge can only be oriented outwards if both red edges are oriented inwards. For an OR gate, any one of the blue edges can be oriented outwards as long as at least one of the other two blue edges is oriented inwards.

\subsection{Gadgets}
\label{sec:gadgets}

Recall that the constraint graph of a nondeterministic constraint logic instance has edges of two types (blue and red), each of which can be oriented in either direction as the instance is reconfigured, and vertices of two types (AND vertices, with two red edges and one blue edge, and OR vertices, with three blue edges) which require at each step that at least one blue edge or two red edges be oriented inwards. The instances of path reconfiguration created by our reduction will have subgraphs (``gadgets'') corresponding to these constraint graph features: edge gadgets corresponding to the edges of the constraint graph (depicted as the blue and red regions of \autoref{fig:reduction}) and vertex gadgets corresponding to the vertices of the constraint graph (the yellow regions of \autoref{fig:reduction}). The edge gadgets will not be differentiated by color (the red and blue edges of the constraint graph generate edge gadgets of the same type) but we will have two different types of vertex gadgets for the two different types of constraint graph vertex. In addition we will also have a third type of gadget, \emph{connection gadgets}, depicted as green in \autoref{fig:reduction}. The edge gadgets of the reduction will be organized into a single cyclic sequence (not necessarily related to the connectivity of the constraint graph) and each pair of consecutive edge gadgets in this cyclic sequence will be connected to each other by a connection gadget.

In more detail:

\begin{itemize}
\item Each vertex gadget has three pairs of entrance and exit vertices (one pair of an entrance and exit vertex for each incident edge gadget) and one or two other interior vertices.
\item The vertex gadgets for OR vertices of the constraint graph (the three on the left side of \autoref{fig:reduction}) have two interior vertices, with each entrance or exit vertex connected to each interior vertex to form a $K_{6,2}$ complete bipartite subgraph.
\item The vertex gadgets for AND vertices of the constraint graph (the three on the right side of \autoref{fig:reduction}) form a tree with six edges. Two of these edges form a path connecting the blue entrance and exit vertices to each other through the single interior vertex. Two more edges connect one pair of a red entrance and exit vertex to each other through the blue entrance vertex,
and the final two edges connect the other pair of a red entrance and exit vertex to each other through the blue exit vertex.
\item Each edge gadget consists of two disjoint paths. One of these two paths connects the two entrance vertices of the two adjacent vertex gadgets and the other connects the two exit vertices of the same two gadgets. Both paths must have equal, even length, but we do not require that these lengths be uniform across all edge gadgets. (In \autoref{fig:reduction}, all edge gadget paths have length 2.)
\item Each connector gadget consists of a single path, from the midpoint of the exit path of one edge gadget to the midpoint of the entrance path of another edge gadget. All paths must have length that is at least as large as a threshold length $\ell$ to be determined later, but they are not required to have uniform lengths. (In \autoref{fig:reduction}, all connector gadget paths have length 3.)
\end{itemize}

\subsection{Valid paths}

Along with its underlying undirected graph, an instance of nondeterministic constraint logic also comes equipped with two orientations of the graph (consistent with its constraints): an \emph{initial orientation} and a \emph{final orientation}. The corresponding objects in our reduction to path reconfiguration will be paths of a certain special form, which we call \emph{valid paths}.

It is convenient to start by defining a \emph{valid cycle} instead of a valid path.
A valid cycle is a cycle $C$ in the path graph that meets the following constraints:
\begin{itemize}
\item $C$ uses all the edges of all the connector gadgets.
\item Within each edge gadget, $C$ uses half of the entrance path, from the midpoint (where it is reached by a connector gadget) to the entrance vertex of one of its two incident vertex gadgets,
and then half of the exit path, from the exit vertex of the same vertex gadget to the midpoint of the exit path (where it reaches another connector gadget).
\item Within each vertex gadget $v$, for each incident edge gadget $e$ such that $C$ follows the paths in $e$ from the midpoints to the entrance and exit of $v$, $C$ connects the entrance and exit vertices for $e$ by a path of length two.
\end{itemize}

All valid cycles have the same length. We define a valid path to be a path that can be formed by removing a single edge from a valid cycle. For instance, the path of heavy black edges in \autoref{fig:reduction} is a valid path.

\begin{observation}
\label{obs:valid-path-completion}
Every valid path has a unique completion to a valid cycle.
\end{observation}

We now define a mapping $\pi$ from valid cycles (or equivalently by \autoref{obs:valid-path-completion}, valid paths) to consistent orientations of the underlying constraint graph. The orientation $\pi(C)$ is determined, for each edge $e$ of the constraint graph, by the following simple rule: orient $e$ \emph{away} from the vertex through which the edge gadget for $e$ is connected.

\begin{lemma}
\label{lem:valid-is-valid}
If $C$ is a valid cycle, $\pi(C)$ is a consistent orientations of the underlying constraint graph.
\end{lemma}

\begin{proof}
In an OR vertex gadget, the only two-edge paths from an entrance to an exit vertex go through one of the interior vertices. There are only two such interior vertices, so there can be only two such paths in $C$, and two edges oriented away from the vertex in $\pi(C)$. Therefore, at least one of the three incident edges must be oriented inwards, meeting the requirements for a consistent orientation at that vertex.

In an AND vertex gadget, the unique two-edge paths from a red entrance vertex to the corresponding red exit vertex pass through an entrance or exit for the blue vertex. Therefore, if either of these paths is used by $C$ (corresponding to a red edge being oriented outwards in $\pi(C)$), the blue two-edge path cannot also be used, and the blue edge must be oriented inwards in $\pi(C)$. Thus, $\pi(C)$ meets the requirements for a consistent orientation at that vertex.
\end{proof}

\begin{lemma}
\label{lem:valid-surjective}
The mapping $\pi$ is a surjective mapping from valid cycles (or equivalently by \autoref{obs:valid-path-completion}, valid paths) to consistent orientations of the underlying constraint graph.
\end{lemma}

\begin{proof}
By \autoref{lem:valid-is-valid}, every valid cycle $C$ corresponds to a consistent orientation $\pi(C)$ so it remains to show that every consistent orientation arises in this way.
That is, given a consistent orientation on the constraint graph, we must construct a valid cycle $C$ such that $\pi(C)$ is that orientation. There is no choice for how $C$ passes through the connector gadgets (it must use all of them), nor through the edge gadgets (it must use the two half-paths on the side away from where the constraint graph edge is directed), so the only remaining part of $C$ to construct is its two-edge paths through each vertex gadget.

In a vertex gadget for an OR vertex of the constraint graph, the consistency of the orientation ensures that $C$ is required to connect at most two pairs of entrance and exit vertices. We can do so greedily, using one of the two interior vertices for a two-edge path for an arbitrarily chosen pair of entrance and exit vertices and then (if there is another pair of entrance and exit vertices to be connected) doing so through the other interior vertex.

In a vertex gadget for an AND vertex of the constraint graph, there is a unique choice of two-edge path for each pair of an entrance and exit vertex that must be connected.
These paths are disjoint unless there is both an outgoing blue edge and an outgoing red edge,
which cannot happen in a consistent orientation of the constraint graph.
\end{proof}

We note, however, that even when considering valid cycles rather than valid paths, $\pi$ is not injective: the same consistent orientation on the constraint graph may be represented by more than one valid cycle. This failure of injectivity will cause complications in our reduction which we handle in \autoref{sec:parity}.

\subsection{Reconfiguration}

We define a \emph{valid reconfiguration move} to be a sequence of path reconfiguration moves that take one valid path to another, according to the following steps:
\begin{itemize}
\item First, slide the path around its valid cycle, so that one of its endpoints is at the vertex where a connector gadget and an edge gadget meet.
\item Then, slide the path along the edge gadget, away from the vertex gadget that it was previously connected through, so that it instead passes through a two-edge path in the opposite vertex gadget.
\end{itemize}
In \autoref{fig:reduction}, the path is already at the vertex where a connector gadget and an edge gadget meet. However, it cannot reconnect through the opposite vertex gadget, because its two-edge path through that gadget is blocked by another part of the path. Therefore, in order to make a valid reconfiguration move in the figure, it would be necessary to first slide the path to end at a different edge gadget.

\begin{lemma}
If $P$ is a valid path, then every valid reconfiguration move corresponds to a consistent reorientation of a single edge in $\pi(P)$ and every consistent reorientation of a single edge in $\pi(P)$ can be realized by a valid reconfiguration move.
\end{lemma}

\begin{proof}
It is straightforward to verify that a valid reconfiguration move changes the orientation of a single edge in $\pi(P)$ The fact that this change of orientation is consistent with the constraints of the constraint graph follows immediately from \autoref{lem:valid-is-valid}.

If $P$ is a valid path, and $e$ is an edge that can be consistently reoriented in $\pi(P)$, then we may reorient $e$ by shifting $P$ along its corresponding valid cycle so that it ends at the edge gadget for $e$ and then following a path through the opposite vertex gadget than the one $P$ previously passed through. If this is the gadget for an AND vertex of the constraint graph, then its path is uniquely determined, and \autoref{lem:valid-surjective} ensures that it is free to be used by $P$.
If it is the gadget for an OR vertex, \autoref{lem:valid-surjective} ensures that at least one of the two interior vertices is available to be used by $P$; if both are available, we can choose which one to use arbitrarily.
\end{proof}

\begin{corollary}
\label{cor:functor}
The mapping $\pi$ can be extended to a surjective map from reconfiguration sequences on the graph of the reduction to reconfiguration sequences on the underlying constraint graph.
\end{corollary}

Again, however, this map is not necessarily injective, an issue that we deal with in the next section.

\subsection{OR-vertex parity}
\label{sec:parity}

When a consistent orientation of the constraint graph has one or two outgoing edges at an OR vertex $v$, the corresponding valid cycles in the graph of the reduction will use one or two two-edge paths through the vertex gadget for $v$.  If there is one two-edge path, it can use either of the two interior vertices of the vertex gadget. If there are two two-edge paths, they can be assigned to the interior vertices in either of two ways. Thus, in either case there are two choices for how the valid cycle passes through the vertex gadget. We refer to this choice as the \emph{parity} of $v$. Because of this freedom, for a consistent orientation in which there are $k$ OR vertices having one or two outgoing edges, there will be $2^k$ valid cycles corresponding to that orientation in the reduction graph.

\begin{figure}[t]
\centering\includegraphics[scale=0.6]{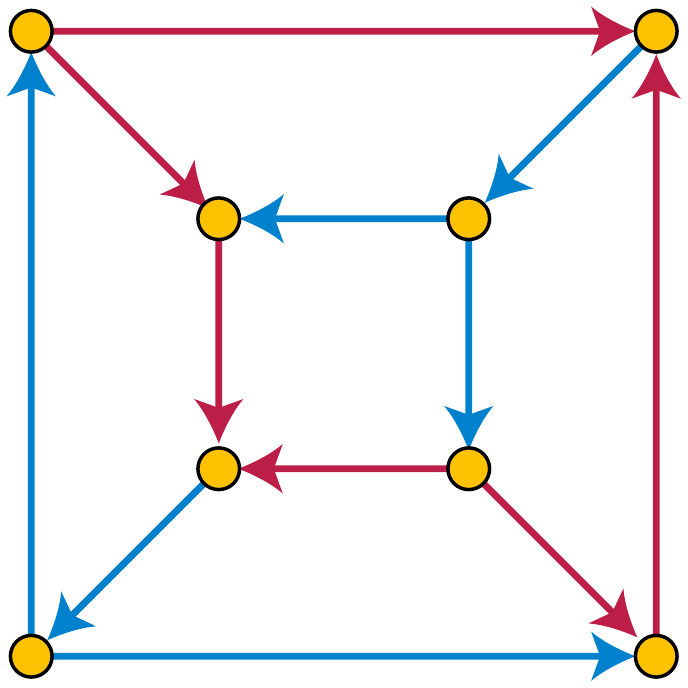}
\caption{An instance of nondeterministic constraint logic in which the OR vertices have two outgoing blue edges and cannot be reconfigured.}
\label{fig:locked-ncl}
\end{figure}

It is not always the case that one of these $2^k$ valid cycles can be reconfigured into a different valid cycle that represents the same orientation of the constraint graph.  \autoref{fig:locked-ncl} depicts an example of this phenomenon. In the constraint graph shown, only the upper left and lower right diagonal red edges can be reoriented; the 10 remaining edges are all locked into the orientations shown in the figure. In particular, the two OR vertices of the figure are both locked in a configuration where they have two outgoing blue edges. In the corresponding valid cycles for the reduction graph,
there are two choices for the parity of the corresponding vertex gadgets. It is not possible to change the parity by using valid reconfiguration moves, because such a move can only change one of the two-edge paths through the vertex gadget at a time, and both two-edge paths would need to be changed simultaneously.

To handle this issue, we choose valid paths for the initial and final orientations of the constraint graph instance with the property that, when the initial and final orientations are both locally the same at any given OR vertex, the corresponding valid paths have the same orientation at that vertex. In this way, if a reconfiguration sequence for the orientations of the constraint graph fails to make any reconfiguration moves at the given OR vertex, the corresponding valid reconfiguration sequence for paths in the reduction graph will never change the parity at that vertex, and will automatically produce a valid path with the correct parity. If, on the other hand, $v$ is an OR vertex of the constraint graph at which the reconfiguration sequence for the orientations of the constraint graph makes a change, then at least one of the orientations in the reconfiguration sequence will have at most one outgoing edge at $v$. If a corresponding reconfiguration sequence for valid paths of the constraint graph produces the wrong parity at $v$, it can be modified by inserting a sequence of moves that change the parity at a time when the valid path passes through $v$ at most once.

The sequence of moves that changes the parity at a vertex is not itself a valid reconfiguration move, because it does not reorient the corresponding constraint graph edge, but it is very similar. We slide the valid path along its valid cycle until it reaches the entrance of the outgoing edge at the corresponding vertex gadget, then slide it through the other choice of two-edge path through the vertex gadget. The change of parity is global: if we insert such a move anywhere that it can be performed within a valid reconfiguration sequence, it will change the parity at that vertex, and only at that vertex, in the valid path reached at the end of the sequence.

Summarizing, we have:

\begin{lemma}
\label{lem:choose-valid-paths}
Given an initial and final orientation of a constraint graph, it is possible to choose corresponding initial and final valid paths in the reduction graph such that there exists a reconfiguration sequence from the initial and final orientations if and only if there exists a sequence of valid reconfiguration moves and parity-change moves that take the initial path to the final path within the reduction graph.
\end{lemma}

\subsection{Deviations from validity}

Unfortunately, not every reconfiguration sequence, starting from a valid path in the reduction graph, consists only of valid reconfiguration moves and parity-change moves. A reconfiguration sequence can deviate from these moves in three ways.

First, it is possible to reverse a valid reconfiguration move or parity-change move. The reverse of a valid reconfiguration move shifts the valid path so that its endpoint is on the exit path (instead of the entrance) and then traverses backwards through an incident vertex gadget. These kinds of variant moves are harmless, because they have the same effect as performing the move in the forward direction.

Second, it is possible for the end of the current path to enter a vertex gadget $v$ at the entrance (or exit) vertex for one incident edge gadget $e$, but then to exit the vertex gadget into a different edge gadget $f$. However, if the connector gadgets are long enough, it will not be able to progress through $f$ to other parts of the reduction graph, because the midpoints of the two paths in $f$ will still be occupied by other parts of the path. Therefore, a reconfiguration sequence that makes this kind of variation cannot use it to reach another valid path; the only possible steps after being blocked in this way are to backtrack back through $v$ to the proper exit or entrance vertex for $e$. A straightforward case analysis shows that, if all edge gadget paths have length at most $2r$, and all connector gadget paths have length at least $r+2$, then this kind of variation will always be blocked. For instance, \autoref{fig:reduction} obeys these constraints with $r=1$.

Third, and most problematically, in an OR vertex gadget with one outgoing edge (such as the upper left vertex gadget of \autoref{fig:reduction}) it is possible for a reconfigured path to connect the entrance and exit vertices of its incident edge gadget by a path of length four instead of by the path of length two that would be required in a valid path. Although doing so would prevent reconfiguration of the other edge gadgets incident to that vertex gadget, it would also have the effect of increasing the length of the gap between the endpoints of the path as measured along the corresponding valid cycle. Potentially, if enough variant reconfigurations of this type could be performed, the gap could be made so large as to encompass multiple vertex gadgets, allowing simultaneous reconfiguration of multiple edges in the corresponding constraint graph. We do not know how to forbid these variant reconfigurations, but we can prevent them from having this harmful effect by making the connector gadget paths even longer, so that (even if some reconfigurations of this type occur) it is still impossible for the gap in the valid path to spread over more than one vertex gadget. In particular, with $r$ as above and $o$ denoting the number of OR vertex gadgets, we let $\ell$ (the minimum connector gadget length) be $2o+r+2$.
With connector gadgets of length at least $\ell$, even a gap of length $2o+1$ between the endpoints of the path (along a corresponding valid cycle) would not be enough to allow a different valid cycle to be reached except via reorientations that could have been performed as valid reconfiguration moves.

Again, we summarize the results of this section in a lemma:

\begin{lemma}
\label{lem:long-enough-gadgets}
If all connector gadgets have length at least $2o+r+2$, where $r$ denotes half of the maximum length of an edge gadget path and $o$ denotes the number of OR vertex gadgets, and if it is possible to reconfigure a valid path $P$ into a valid path $Q$, then it is possible to do so using only valid reconfiguration moves and parity-change moves.
\end{lemma}

\subsection{Bandwidth}

A \emph{linear layout} of an $n$-vertex graph is a labeling of its vertices with distinct integers in the range from 1 to $n$. The \emph{width} of a layout is the maximum difference between the labels of the two endpoints of any edge, and the \emph{bandwidth} of the graph is the minimum width of any of its layouts. Although bandwidth is hard to compute, and hard to approximate even for some very simple graphs~\cite{DubFeiUng-JCSS-11}, this will not be problematic for us, as (from the previous hardness reduction for low-bandwidth nondeterministic constraint logic~\cite{vdZ-IPEC-15}) we can assume that the input to our reduction is a constraint graph with an already-given low-bandwidth layout. As we show in this section, there is enough remaining flexibility in our reduction (in its choice of path lengths and the ordering to connect the edge gadgets via the connector gadgets) that we can preserve the bounded bandwidth of its instances.

\begin{lemma}
\label{lem:bandwidth-stays-low}
The reduction described above can be implemented so that, for constraint graphs of bounded bandwidth with a given bounded-width layout, we can produce reduction graphs of bounded bandwidth.
\end{lemma}

\begin{proof}
We will assign non-integer numeric labels to the vertices of the reduction graph, allowing some vertices to have equal labels. We then perturb these labels to produce a total ordering of all the vertices, as follows.

To begin with, let $n$ be the number of vertices in the constraint graph. We set the length of all edge gadget paths to $2n$, and we set the length of all connector gadget paths to $3n+2$, meeting the requirements of \autoref{lem:long-enough-gadgets}. Let $G'$ be the graph formed from the constraint graph by replacing each edge by a two-edge path; if $G$ has bounded bandwidth, so does $G'$, and we can construct a bounded-width linear layout of $G'$ from the layout of $G$.
We label each vertex in a vertex gadget of the reduction graph by the integer label of the corresponding constraint graph vertex in $G'$. We label the midpoint of each edge gadget path by the integer label of the subdivision vertex on its corresponding two-edge path.
We label the remaining interior vertices of the paths of each edge gadget in such a way that they are evenly spaced between the labels of the constraint graph vertex and subdivision vertex.

This labeling also gives us a total ordering of the midpoints of the edge gadget paths. We use this total ordering to choose how to connect the edge gadgets by a cycle of connector gadgets. This cycle will consist of two monotone paths, $X$ and $Y$, from the leftmost edge gadget midpoint to the rightmost edge gadget midpoint. Each edge gadget between the leftmost and rightmost is assigned to one of these two paths, so that the assignment alternates between $X$ and $Y$ in the total ordering of edge gadget midpoints. We then connect each consecutive pair of edge gadgets in either $X$ or $Y$ by a connector gadget.
We label the interior vertices of each connector gadget in such a way that they are evenly spaced between the labels of the two incident edge gadget midpoints.

If the given graph has bounded bandwidth, then each edge gadget path or connector gadget path will have endpoints whose labels differ by a constant. Therefore, the difference between the labels of any two adjacent vertices within these gadgets is $O(1/n)$, and the difference between the labels of any two adjacent vertices within a vertex gadget is of course zero. However, for any interval $I$ of the real number line of length $O(1/n)$, our construction will assign only $O(1)$ vertices of the reduction graph to have labels within $I$. Therefore, if the labeling is arbitrarily perturbed to produce a total ordering, and the ranks in this total ordering are used as a linear layout of the reduction graph, the bandwidth of the resulting graph will be bounded.
\end{proof}

\subsection{Hardness}

Combining the results above, we have:

\begin{theorem}
The path reconfiguration decision problem is $\mathsf{PSPACE}$-complete, even for graphs of bounded bandwidth.
\end{theorem}

\begin{proof}
Membership in $\mathsf{PSPACE}$ is clear. To prove completeness, we reduce from nondeterministic constraint logic in graphs of bounded bandwidth, with a given low-bandwidth layout; this was previously known to be $\mathsf{PSPACE}$-complete~\cite{vdZ-IPEC-15}). The gadgets of the reduction are those given in \autoref{sec:gadgets}, with the detailed choice of path lengths and connector gadget ordering of \autoref{lem:bandwidth-stays-low}. The initial and final orientations of the given constraint graph are transformed into start and goal paths of a path reconfiguration problem according to \autoref{lem:choose-valid-paths}.

It is straightforward to perform the reduction in polynomial time, and by \autoref{lem:bandwidth-stays-low}, the bandwidth of the resulting reduction graph is bounded. If a solution exists to the given nondeterministic constraint logic problem, then one exists for the corresponding path reconfiguration problem by \autoref{lem:choose-valid-paths}. If no solution exists for the constraint logic problem,
then correspondingly no solution exists for the path reconfiguration problem by \autoref{lem:long-enough-gadgets} and \autoref{cor:functor}.
\end{proof}

\bibliographystyle{amsplainurl}
\bibliography{snake}

\providecommand{\bysame}{\leavevmode\hbox to3em{\hrulefill}\thinspace}
\providecommand{\MR}{\relax\ifhmode\unskip\space\fi MR }
\providecommand{\MRhref}[2]{%
  \href{http://www.ams.org/mathscinet-getitem?mr=#1}{#2}%
}
\providecommand{\href}[2]{#2}
\begin{thebibliography}{10}

\bibitem{AloYusZwi-JACM-95}
Noga Alon, Raphael Yuster, and Uri Zwick, \emph{{Color-coding}}, Journal of the
  ACM \textbf{42} (1995), no.~4, 844{--}856, \href
  {http://dx.doi.org/10.1145/210332.210337} {\path{doi:10.1145/210332.210337}},
  \MR{1411787}.

\bibitem{Bod-Algs-93}
Hans~L. Bodlaender, \emph{{On linear time minor tests with depth-first
  search}}, Journal of Algorithms \textbf{14} (1993), no.~1, 1{--}23, \href
  {http://dx.doi.org/10.1006/jagm.1993.1001}
  {\path{doi:10.1006/jagm.1993.1001}}, \MR{1199244}.

\bibitem{Bon-TCS-13}
Paul Bonsma, \emph{{The complexity of rerouting shortest paths}}, Theoretical
  Computer Science \textbf{510} (2013), 1{--}12, \href
  {http://dx.doi.org/10.1016/j.tcs.2013.09.012}
  {\path{doi:10.1016/j.tcs.2013.09.012}}, \MR{3122210}.

\bibitem{DeBOph-FUN-16}
Marzio De~Biasi and Tim Ophelders, \emph{{The complexity of Snake}}, 8th
  International Conference on Fun with Algorithms, FUN 2016, June 8{--}10,
  2016, La Maddalena, Italy (Erik~D. Demaine and Fabrizio Grandoni, eds.),
  Leibniz International Proceedings in Informatics (LIPIcs), vol.~49, Schloss
  Dagstuhl {--} Leibniz-Zentrum f{\"u}r Informatik, 2016, pp.~11:1{--}11:13,
  \href {http://dx.doi.org/10.4230/LIPIcs.FUN.2016.11}
  {\path{doi:10.4230/LIPIcs.FUN.2016.11}}.

\bibitem{DubFeiUng-JCSS-11}
Chandan Dubey, Uriel Feige, and Walter Unger, \emph{{Hardness results for
  approximating the bandwidth}}, Journal of Computer and System Sciences
  \textbf{77} (2011), no.~1, 62{--}90, \href
  {http://dx.doi.org/10.1016/j.jcss.2010.06.006}
  {\path{doi:10.1016/j.jcss.2010.06.006}}, \MR{2767125}.

\bibitem{FelLan-JCSS-94}
Michael~R. Fellows and Michael~A. Langston, \emph{{On search, decision, and the
  efficiency of polynomial-time algorithms}}, Journal of Computer and System
  Sciences \textbf{49} (1994), no.~3, 769{--}779, \href
  {http://dx.doi.org/10.1016/S0022-0000(05)80079-0}
  {\path{doi:10.1016/S0022-0000(05)80079-0}}, \MR{1306142}.

\bibitem{GupSaaZeh-arXiv-19}
Siddharth Gupta, Guy Sa'ar, and Meirav Zehavi, \emph{{The parameterized
  complexity of motion planning for snake-like robots}}, Electronic preprint
  arxiv:1903.02445, March 2019.

\bibitem{HanItoMiz-COCOON-18}
Tesshu Hanaka, Takehiro Ito, Haruka Mizuta, Benjamin Moore, Naomi Nishimura,
  Vijay Subramanya, Akira Suzuki, and Krishna Vaidyanathan,
  \emph{{Reconfiguring spanning and induced subgraphs}}, Computing and
  Combinatorics: 24th International Conference, COCOON 2018, Qing Dao, China,
  July 2{--}4, 2018, Proceedings (Lusheng Wang and Daming Zhu, eds.), Lecture
  Notes in Computer Science, vol. 10976, Springer, 2018, pp.~428{--}440, \href
  {http://arxiv.org/abs/1803.06074} {\path{arXiv:1803.06074}}, \href
  {http://dx.doi.org/10.1007/978-3-319-94776-1_36}
  {\path{doi:10.1007/978-3-319-94776-1_36}}.

\bibitem{HeaDem-ICALP-02}
Robert~A. Hearn and Erik~D. Demaine, \emph{{The nondeterministic constraint
  logic model of computation: reductions and applications}}, Automata,
  Languages and Programming: 29th International Colloquium, ICALP 2002,
  M{\'a}laga, Spain, July 8{--}13, 2002 Proceedings, Lecture Notes in Computer
  Science, vol. 2380, Springer, 2002, pp.~401{--}413, \href
  {http://dx.doi.org/10.1007/3-540-45465-9_35}
  {\path{doi:10.1007/3-540-45465-9_35}}, \MR{2062475}.

\bibitem{HeaDem-09}
\bysame, \emph{{Games, Puzzles, and Computation}}, A K Peters, 2009.

\bibitem{HopWon-STOC-74}
J.~E. Hopcroft and J.~K. Wong, \emph{{Linear time algorithm for isomorphism of
  planar graphs (preliminary report)}}, Proceedings of the Sixth Annual ACM
  Symposium on Theory of Computing (STOC '74) (1974), 172{--}184, \href
  {http://dx.doi.org/10.1145/800119.803896} {\path{doi:10.1145/800119.803896}},
  \MR{0433964}.

\bibitem{ItoDemHar-TCS-11}
Takehiro Ito, Erik~D. Demaine, Nicholas J.~A. Harvey, Christos~H.
  Papadimitriou, Martha Sideri, Ryuhei Uehara, and Yushi Uno, \emph{{On the
  complexity of reconfiguration problems}}, Theoretical Computer Science
  \textbf{412} (2011), no.~12-14, 1054{--}1065, \href
  {http://dx.doi.org/10.1016/j.tcs.2010.12.005}
  {\path{doi:10.1016/j.tcs.2010.12.005}}, \MR{2797748}.

\bibitem{KamMedMil-TCS-11}
Marcin Kami{\'n}ski, Paul Medvedev, and Martin Milani{\v{c}}, \emph{{Shortest
  paths between shortest paths}}, Theoretical Computer Science \textbf{412}
  (2011), no.~39, 5205{--}5210, \href
  {http://dx.doi.org/10.1016/j.tcs.2011.05.021}
  {\path{doi:10.1016/j.tcs.2011.05.021}}, \MR{2857671}.

\bibitem{MouNisRam-Algo-17}
Amer~E. Mouawad, Naomi Nishimura, Venkatesh Raman, Narges Simjour, and Akira
  Suzuki, \emph{{On the parameterized complexity of reconfiguration problems}},
  Algorithmica \textbf{78} (2017), no.~1, 274{--}297, \href
  {http://dx.doi.org/10.1007/s00453-016-0159-2}
  {\path{doi:10.1007/s00453-016-0159-2}}, \MR{3620830}.

\bibitem{NesOss-12}
Jaroslav Ne{\v{s}}et{\v{r}}il and Patrice Ossona~de Mendez, \emph{{Chapter 6.
  Bounded height trees and tree-depth}}, Sparsity: Graphs, Structures, and
  Algorithms, Algorithms and Combinatorics, vol.~28, Springer, 2012,
  pp.~115{--}144, \href {http://dx.doi.org/10.1007/978-3-642-27875-4}
  {\path{doi:10.1007/978-3-642-27875-4}}, \MR{2920058}.

\bibitem{vdZ-IPEC-15}
Tom~C. van~der Zanden, \emph{{Parameterized complexity of graph constraint
  logic}}, 10th International Symposium on Parameterized and Exact Computation
  (IPEC 2015) (Thore Husfeldt and Iyad Kanj, eds.), Leibniz International
  Proceedings in Informatics (LIPIcs), vol.~43, Schloss Dagstuhl {--}
  Leibniz-Zentrum f{\"u}r Informatik, 2015, pp.~282{--}293, \href
  {http://dx.doi.org/10.4230/LIPIcs.IPEC.2015.282}
  {\path{doi:10.4230/LIPIcs.IPEC.2015.282}}.

\bibitem{Wro-JCSS-18}
Marcin Wrochna, \emph{{Reconfiguration in bounded bandwidth and tree-depth}},
  Journal of Computer and System Sciences \textbf{93} (2018), 1{--}10, \href
  {http://dx.doi.org/10.1016/j.jcss.2017.11.003}
  {\path{doi:10.1016/j.jcss.2017.11.003}}, \MR{3736899}.

\end{thebibliography}

\end{document}